\documentclass[%
twocolumn, 
notitlepage, 
aps,
prb,
10pt, 
superscriptaddress,
floatfix,
letterpaper,
reprint,
nofootinbib
]{revtex4-2}

\usepackage{amsmath,amssymb,amsthm}
\usepackage{mathtools}
\usepackage{graphicx}
\usepackage{dcolumn}
\usepackage{bm}
\usepackage[compat=1.1.0]{tikz-feynman}
\tikzfeynmanset{warn luatex=false}
\usepackage{physics}
\usepackage{hyperref}

\definecolor{linkcolor}{HTML}{399B03}
\definecolor{urlcolor}{HTML}{399B03}
\hypersetup{pdfstartview=FitH, linkcolor=linkcolor,urlcolor=urlcolor,colorlinks=true}
\hypersetup{
    unicode=false,          
    pdftoolbar=true,        
    pdfmenubar=true,        
    pdffitwindow=false,     
    pdfstartview={FitH},    
    pdfauthor={},           
    colorlinks=true,        
    linkcolor=blue,         
    citecolor=red,          
}

\newcommand{\OO}[1]{\mathcal{O}(#1)}

\newtheorem{remark}{Remark}
\newtheorem{lemma}{Lemma}
\newtheorem{theorem}{Theorem}

\newcommand{\be}{\begin{equation}}
\newcommand{\ee}{\end{equation}}
\newcommand{\ba}{\begin{aligned}}
\newcommand{\ea}{\end{aligned}}

\newcommand{\x}{\mathbf{r}}

\newcommand{\CCQ}{Center for Computational Quantum Physics, Flatiron Institute, 162 5th Avenue, New York, NY 10010, USA}
\newcommand{\CCM}{Center for Computational Mathematics, Flatiron Institute, 162 5th Avenue, New York, NY 10010, USA}
\newcommand{\courant}{Courant Institute of Mathematical Sciences, New York University, New York, NY 10012, USA}
\newcommand{\westlake}{Institute for Theoretical Sciences, Westlake University, Hangzhou, Zhejiang 310030, China}

\begin{document}

\title{Interpolative separable density fitting on adaptive real space grids}

\author{Hai Zhu}
\thanks{These authors contributed equally to this work.}
\affiliation{\westlake}

\author{Chia-Nan Yeh}
\thanks{These authors contributed equally to this work.}
\affiliation{\CCQ}

\author{Miguel A. Morales}
\affiliation{\CCQ}

\author{Leslie Greengard}
\affiliation{\CCM}
\affiliation{\courant}

\author{Shidong Jiang}
\thanks{Corresponding authors (sjiang@flatironinstitute.org, jkaye@flatironinstitute.org)}
\affiliation{\CCM}

\author{Jason Kaye}
\thanks{Corresponding authors (sjiang@flatironinstitute.org, jkaye@flatironinstitute.org)}
\affiliation{\CCQ}
\affiliation{\CCM}

\begin{abstract}
We generalize the interpolative separable density fitting (ISDF) method, used for compressing the four-index electron repulsion integral (ERI) tensor, to incorporate adaptive real space grids for potentially highly localized single-particle basis functions.
To do so, we employ a fast adaptive algorithm, the recently-introduced dual-space multilevel kernel-splitting method, to solve the Poisson equation for the ISDF auxiliary basis functions. 
The adaptive grids are generated using a high-order accurate, black-box procedure that satisfies a user-specified error tolerance. 
Our algorithm relies on the observation, which we prove, that an adaptive grid resolving the pair densities appearing in the ERI tensor can be straightforwardly constructed from one that resolves the single-particle basis functions, with the number of required grid points differing only by a constant factor. 
We find that the ISDF compression efficiency for the ERI tensor with highly localized basis sets is comparable to that for smoother basis sets compatible with uniform grids. 
To demonstrate the performance of our procedure, we consider several molecular systems with all-electron basis sets which are intractable using uniform grid-based methods. 
Our work establishes a pathway for scalable many-body electronic structure simulations with arbitrary smooth basis functions, making simulations of phenomena like core-level excitations feasible on a large scale.
\end{abstract}

\maketitle

\section{Introduction \label{sec:intro}}

The four-index electron repulsion integral (ERI) tensor is a fundamental
building block of electronic structure theories, representing the Coulomb interactions between products of $N$ single-particle basis functions $\phi_{i}(\mathbf{r})$: 
\begin{align}
V_{ijkl} = \int d\mathbf{r} \int d\mathbf{r}' \phi_{i}(\mathbf{r})\phi_{j}(\mathbf{r})\frac{1}{|\mathbf{r}-\mathbf{r}'|}\phi_{k}(\mathbf{r}')\phi_{l}(\mathbf{r}').
\label{eq:eri}
\end{align}
The ERI tensor serves as a universal starting point for incorporating quantum many-body effects within an electronic structure simulation based on a single-particle basis. 
However, building and storing this tensor directly is typically infeasible for large system sizes, and can become a primary bottleneck: it amounts to solving the Poisson equation for $N^2$ orbital pairs, computing an inner product for all $N^4$ orbital combinations, and storing $N^4$ tensor elements. 
Assuming the number of grid points scales as $\OO{N}$, this yields an $\OO{N^5}$ total cost.

Even though at first sight the problem of computing and storing $V_{ijkl}$ seems too daunting for large problems, the task is made more manageable when we realize that the tensor is expected to be significantly rank deficient. 
Here, we show that the rank of $V_{ijkl}$, interpreted as an $N^2 \times N^2$ matrix, scales linearly with the basis size $N$. As a result, the direct evaluation of \eqref{eq:eri} is typically replaced by more efficient formulations that lead to compact representations with reduced computational and storage requirements.
Common approaches include resolution of identity (RI)~\cite{DF_Werner2003,CD_DF_Weigend2009,DF_Ren2012,GDF_MDF_Sun2017,RSDF_HongZhou2021}, Cholesky decomposition~\cite{Cholesky_ERI_MOL_Beebe1977,CD_Koch2003,CD_DF_Weigend2009}, tensor hypercontraction (THC)~\cite{THC_one_Martinez2012,LSTHC_Sherrill2012,ISDF_Lu2015,ISDF_Bloch_Lu2016,THC-RPA_CNY2023}, and the canonical polyadic (CP) approximation~\cite{CP_Chinnamsetty2007,CP_Benedikt2013,CP2_Benedikt2013,CP_Bohm2016,CPCC_Schutski2017}. 
Other approaches include hierarchical matrix representations of the ERI~\cite{chow2020,Xing_JCP2020}, sparse representations~\cite{pan25}, and tensor network methods to represent the basis functions~\cite{Nicolas_PRB2025}.
Each of these schemes involves a particular compression of the ERI tensor, with different trade-offs in terms of computational cost, accuracy, and applicability.

The most common approach used to compress the ERI tensor in quantum many-body methods, particularly in the quantum chemistry community, is the RI decomposition. It has a good balance of efficiency, accuracy and robustness, which has led to widespread use and quick adoption. It, nonetheless, suffers from serious limitations since it requires $\OO{N^3}$ storage and leads to $\OO{N^4}$ scaling algorithms when used in correlated many-body calculations. In order to reduce memory requirements and computational costs, more aggressive factorizations are needed. Among competing alternatives, THC stands out for enabling compact ($\OO{N^2}$ storage), low-scaling, black-box implementations of a broad range of electronic structure methods, including hybrid density functional theory (DFT)~\cite{ISDF_QRCP_hybrid_Hu2017,ISDF_CVT_hybrid_Lin2018,ISDF_hybridDFT_NAO_Qin2020,ISDF_CVT_hybrid_NAO_Qin2020}, Hartree-Fock (HF)~\cite{rPS_HF_Sharma2022}, coupled-cluster theories~\cite{THC_three_Martinez2012,THC_CC2_Hohenstein2013,THC_EOMCC2_Hohenstein2013,LSTHC_CCSD_Parrish2014,THC_CC_Hohenstein2022}, many-body perturbation theory (e.g., MP2, MP3, $GW$)~\cite{THC_one_Martinez2012,LSTHC_Sherrill2012,THC_DVR_Sherrill2013,LSTHC_MP2_Schumacher2015,THC_SOS_MP2_one_Song2016,THC_SOS_MP2_two_Song2017,THC_SOS_MP2_gradient_Song2017,THC_MP3_Joonho2020,THC_MP3_Matthews2021,TDDFT_G0W0_ISDF_Gao2020,G0W0_COHSEX_ISDF_Ma2021,separable_RI_G0W0_Blase2021}, and auxiliary-field quantum Monte Carlo (AFQMC)~\cite{AFQMC_ISDF_Miguel2019}. One of the leading methods of constructing the THC representation is the interpolative separable density fitting (ISDF) algorithm~\cite{ISDF_Lu2015,ISDF_Bloch_Lu2016}, which scales as $\OO{N^3}$, offers controllable accuracy, and relies on standard numerical linear algebra routines. Other methods of constructing the THC representation, such as least-squares (LS) THC~\cite{LSTHC_Sherrill2012,THC_DVR_Sherrill2013,LSTHC_MP2_Schumacher2015}, have $\OO{N^4}$ computational complexity and often require an initial RI/Cholesky decomposition. 
For large systems, particularly in periodic calculations with a large number of k-points, the generation of the RI/Cholesky decomposition is typically unfeasible, limiting the utility of the LS-THC approach.  

ISDF uses the interpolative decomposition~\cite{ID_Cheng2005,ID_Liberty2007} to construct a low rank representation of the family of pair densities $\rho_{ij}(\mathbf{r}) = \phi_i(\mathbf{r})\phi_j(\mathbf{r})$, in the form of an auxiliary basis $\zeta_{\mu}(\mathbf{r})$. From these, one can construct a projection of the Coulomb integral operator, 
\begin{equation}
    V_{\mu \nu} = \int d\mathbf{r} \int d\mathbf{r}' \zeta_{\mu}(\mathbf{r})\frac{1}{|\mathbf{r}-\mathbf{r}'|}\zeta_{\nu}(\mathbf{r}'),
\end{equation}
which directly leads to a THC decomposition of $V_{ijkl}$. Sec. \ref{subsec:isdf} provides a more detailed description of this procedure. Thus, building the THC decomposition of $V_{ijkl}$ using ISDF requires two main steps: (1) build the auxiliary basis $\zeta_{\mu}(\mathbf{r})$ from the pair densities $\rho_{ij}(\mathbf{r})$, and (2) compute the Coulomb integrals $V_{\mu\nu}$. The first step is agnostic to the real space grid on which the pair densities are represented. The second step is the solution of the Poisson equation for the auxiliary basis functions $\zeta_{\mu}(\mathbf{r})$.

In previous works, the THC-ISDF approach has been applied to periodic systems, using a uniform real space grid to represent the pair densities and a standard fast Fourier transform (FFT)-based Poisson solver. This limits the method to single-particle basis functions which are well-resolved on uniform grids with a modest number of points, requiring the use of pseudopotentials. 
However, many methods, such as the fast multipole method and multi-grid based approaches, have been developed to solve the Poisson equation on adaptive grids which automatically refine into localized features of the density.
For the present setting---solving the Poisson equation on a cubic domain---the state of the art algorithms are black-box, linear scaling, and capable of delivering user-specified accuracy with computational throughput (in grid points per second) approaching that of the FFT. 
Following this approach, we build an adaptive grid resolving the collection of pair densities, use ISDF to construct an auxiliary basis on this grid, and use the recently introduced dual-space multilevel kernel-splitting (DMK)~\cite{jiang2025cpam} algorithm as an adaptive Poisson solver to compute $V_{\mu\nu}$. 
This yields a black-box, $\OO{N^3}$ algorithm to build a THC decomposition of the ERI tensor which allows for general, potentially highly-localized single-particle basis functions. 
While our focus in this work is on molecular systems, our framework is equally applicable to periodic systems by imposing periodic boundary conditions in the DMK solver. 

The ability to operate on adaptive grids with no assumptions on the analytical form of the basis functions makes ISDF applicable to a wide variety of single-particle basis sets beyond Gaussians, including representations like projector-augmented waves (PAW) and linearized augmented plane waves (LAPW).
It enables a cubic-scaling THC construction for all-electron calculations, in which incorporating localized core basis functions is essential for an accurate description of core excitations.

The remainder of this paper is organized as follows. Section~\ref{sec:THC} describes THC, ISDF, and our adaptive grid approach. In Section~\ref{sec:results}, we evaluate the performance and accuracy of the method across a range of basis sets, chemical species, and system sizes, and analyze its effectiveness for downstream electronic structure predictions. We conclude with a summary and outlook in Section~\ref{sec:conclusion}.

\section{THC decomposition via ISDF using adaptive grids \label{sec:THC}}
The tensor hypercontraction (THC) decomposition of the ERI tensor takes the form 
\begin{align}
V_{ijkl} \approx \sum_{\mu,\nu=1}^R X_{i\mu} X_{j\mu} V_{\mu\nu} X_{k\nu} X_{l\nu},
\label{eq:eri_thc}
\end{align}
for $V_{\mu \nu}$ a projection of the original tensor, and $X_{i \mu}$ the collocation matrix. The rank $R$ of the decomposition will be discussed below. Thus the THC decomposition separates the indices $i,j,k$, and $l$ by introducing two auxiliary indices $\mu$ and $\nu$.

The THC form in \eqref{eq:eri_thc} enables contractions over the ERI tensor to be reorganized into a sequence of tractable matrix–matrix multiplications.
This can already be seen from the exchange potential $K$ in the Hartree-Fock approximation: 
\begin{subequations}
\begin{align}
K_{ij} &= -\sum_{a,b=1}^N \rho_{ab}V_{iabj} \label{eq:thc_Ka}\\
&= -\sum_{\mu,\nu = 1}^R X_{i\mu} V_{\mu\nu} X_{j\nu} \sum_{a,b=1}^N X_{a\mu}\rho_{ab}X_{b\nu}. \label{eq:thc_Kb}
\end{align}
\label{eq:thc_K}%
\end{subequations}
Here $\rho$ is the single-particle density matrix. The complete separation of the orbital indices in \eqref{eq:thc_Kb} reduces the cost of computing \eqref{eq:thc_Ka} from $\OO{N^4}$ to $\OO{R N^2 + R^2 N}$. 

A similar advantage arises in the $GW$ approximation, in which the dynamic self-energy $\Sigma(\tau)$ involves a similar contraction: 
\begin{subequations}
\begin{align}
\Sigma_{ij}(\tau) &= -\sum_{a,b=1}^N G_{ab}(\tau)W_{iabj}(\tau) \label{eq:thc_GWa}\\
&= -\sum_{\mu,\nu=1}^R X_{i\mu}W_{\mu\nu}(\tau)X_{j\nu} \sum_{a,b=1}^N X_{a\mu}G_{ab}(\tau)X_{b\nu}. \label{eq:thc_GWb}
\end{align}
\label{eq:thc_GW}%
\end{subequations}
Here $G(\tau)$ is the imaginary time Green's function, and $W(\tau)$ is the screened interaction.
We refer to Ref.~\onlinecite{Yeh_ISDF_GW2024} for further details on the application of THC to both HF and $GW$. 
 
\subsection{Interpolative separable density fitting\label{subsec:isdf}}

ISDF is an algorithm which constructs the THC decomposition \eqref{eq:eri_thc}. We begin with the pair density
\begin{equation}
\rho_{ij}(\mathbf{r}) = \phi_{i}(\mathbf{r})\phi_{j}(\mathbf{r}), 
\end{equation}
for $i,j = 1,\ldots,N$. Although there are $N^2$ pair densities, 
one might hope that most are numerically linearly dependent, i.e., that the numerical rank of the set of functions $\rho_{ij}(\mathbf{r})$ is $R \ll N^2$. 
This implies the existence of a decomposition of the form
\begin{equation}
\rho_{ij}(\mathbf{r}) \approx \sum_{\mu=1}^R \phi_{i}(\textbf{r}_{\mu})\phi_{j}(\textbf{r}_{\mu})\zeta_{\mu}(\textbf{r}).
\label{eq:id_rho}
\end{equation}
To see this, we consider $\rho_{ij}(\mathbf{r})$ as an $\infty \times N^2$ matrix of rank $R$ (or, if preferred, an $M \times N^2$ matrix with large $M$, using a dense discretization of $\mathbf{r}$ by a grid of $M$ points), and note that the row rank is equal to the column rank.
In other words, $\rho_{ij}(\mathbf{r})$ can be reconstructed from a linear combination of $R$ of its rows, with coefficients $\zeta_\mu(\mathbf{r})$. We refer to $\textbf{r}_{\mu}$ as the interpolating points, and $\zeta_{\mu}(\textbf{r})$ as the auxiliary basis.

If one indeed discretizes $\rho_{ij}(\mathbf{r})$ on a grid of $M$ points to obtain an $M \times N^2$ matrix, then \eqref{eq:id_rho} constitutes an interpolative decomposition (ID) of that matrix \cite{ID_Cheng2005,ID_Liberty2007}. Several algorithms exist to construct this ID~\cite{ISDF_Lu2015,ISDF_CVT_hybrid_Lin2018,LSTHC_Matthews2020,THC-RPA_CNY2023}. In this work we adopt the approach described in Refs.~\onlinecite{LSTHC_Matthews2020,THC-RPA_CNY2023}, summarized in App.~\ref{app:isdf}, which employs a pivoted Cholesky decomposition of the square of the pair density matrix. This yields the decomposition \eqref{eq:id_rho}, with the rank $R$ chosen to meet a desired approximation accuracy. 

As is discussed in Sec.~\ref{subsec:adapt_poisson} and App.~\ref{app:pair_density_error}, the size $M$ of the real space grid required to resolve all $N^2$ pair densities differs by a constant factor (independent of $N$) from that required to resolve all $N$ single-particle basis functions, so $M = \OO{N}$. 
Since $R \leq \min(M, N^2)$, we have $R = \OO{N}$, yielding a compression of $\rho_{ij}(\mathbf{r})$ to $\OO{N}$ auxiliary basis functions. The compression factor will be explored numerically in Sec.~\ref{sec:results}.

A THC decomposition of the ERI tensor can be immediately obtained from the compressed representation \eqref{eq:id_rho}: 
\begin{widetext}
\begin{equation}
\begin{aligned}
V_{ijkl} &= \int d\mathbf{r} \int d\mathbf{r}' \rho_{ij}(\mathbf{r})\frac{1}{|\mathbf{r}-\mathbf{r}'|}\rho_{kl}(\mathbf{r}')\\
&\approx \sum_{\mu,\nu=1}^R \phi_{i}(\mathbf{r}_{\mu})\phi_{j}(\textbf{r}_{\mu}) \Big[\int d\textbf{r}  \int d\textbf{r}' \zeta_{\mu}(\textbf{r})\frac{1}{|\textbf{r}-\textbf{r}'|}\zeta_{\nu}(\textbf{r}') \Big] \phi_{k}(\textbf{r}_{\nu})\phi_{l}(\textbf{r}_{\nu})\\
&= \sum_{\mu,\nu=1}^R X_{i\mu}X_{j\mu}V_{\mu\nu}X_{k\nu}X_{l\nu},
\end{aligned}
\label{eq:eri_isdf}
\end{equation}
\end{widetext}
where 
\begin{subequations}
\begin{align}
&X_{i\mu} = \phi_{i}(\textbf{r}_{\mu}), \label{eq:X_isdf} \\ 
&V_{\mu\nu} = \int d\textbf{r} \int d\textbf{r}' \zeta_{\mu}(\textbf{r})\frac{1}{|\textbf{r}-\textbf{r}'|}\zeta_{\nu}(\textbf{r}').  \label{eq:V_isdf}
\end{align}
\end{subequations}
The accuracy of this decomposition is controlled by the accuracy of the ID approximation of the pair densities, which can be systematically improved by increasing $R$. Indeed, if $R = N^2$, then the ID can be made exact, yielding a trivial THC decomposition for which the projected Coulomb matrix contains $R^2 = N^4$ degrees of freedom. The memory requirement of the THC decomposition is dominated by that of the projected Coulomb matrix $V_{\mu\nu}$, and therefore scales as $\OO{R^2}$. Thus we obtain a quadratic scaling storage requirement, an improvement over the quartic scaling of the original ERI tensor, and the cubic scaling of alternative methods like RI and Cholesky decomposition.

The key steps of the ISDF procedure are therefore (1) the resolution of the pair densities $\rho_{ij}(\mathbf{r})$ on a real space grid of $M$ points, (2) the ID of the resulting $M \times N^2$ matrix to obtain the auxiliary basis $\zeta_{\mu}(\mathbf{r})$ on the real space grid, and (3) the solution of the Poisson equation for the auxiliary basis functions, and subsequent inner products, to obtain the projected Coulomb matrix $V_{\mu\nu}$ in Eq.~\ref{eq:V_isdf}. As described in App.~\ref{app:pair_density_error}, the ID step is a black-box linear algebra procedure that is agnostic to the real space grid on which the pair densities are represented. We therefore discuss Steps 1 and 3.

\subsection{Adaptive real space grids and the solution of the Poisson equation \label{subsec:adapt_poisson}}

Evaluating the integrals \eqref{eq:V_isdf} requires (i) solving the collection of Poisson equations $-\Delta u_{\nu}(\mathbf{r}) = \zeta_{\nu}(\mathbf{r})$, and (ii) computing the inner products $\int d\mathbf{r} \, \zeta_{\mu}(\mathbf{r}) u_{\nu}(\mathbf{r})$. The auxiliary basis functions $\zeta_{\mu}(\mathbf{r})$ are given on a real space grid of $M$ points. 
For simulations in which core electrons are excluded, e.g., using pseudopotentials or effective core potentials, the single-particle basis functions, and therefore the auxiliary basis functions, tend to be smooth even near nuclei. In this case, a uniform real space grid is sufficient, and the Poisson equation can be solved using the fast Fourier transform (FFT). Then the inner products can be computed using a suitable uniform grid quadrature rule.

\begin{remark}
  For periodic systems, the single-particle basis functions and therefore the auxiliary basis functions are considered to be periodic, and the Poisson equation can be solved on a single unit cell with periodic boundary conditions under the charge neutrality assumption, i.e., the mean of the electron density is zero. The inner products can then be computed by the trapezoid rule with spectral accuracy \cite{trefethen2014sirev}. For molecular systems, we assume that all single-particle basis functions have decayed sufficiently by the boundary of the simulation domain, and we solve the Poisson equation in free space. The boundary is determined by gradually enlarging the computational domain, with an a posteriori check based on the maximum normalized $L^2$ norm of the single-particle basis functions evaluated on a spherical boundary. Straightforward use of the FFT to solve the Poisson equation in this setting is inappropriate, since it assumes periodic boundary conditions; in other words, it will introduce slowly-decaying spurious periodic images of the desired solution. Rather, to use uniform grids in a non-periodic setting, the FFT should be used in conjunction with a truncated kernel method, such as that of Ref.~\cite{vico2016jcp}, which correctly imposes free space boundary conditions. In the present work, we indeed consider the molecular case, but use adaptive rather than uniform grids, so this discussion is not relevant.
\end{remark}

In practice, the use of uniform grids is a significant limitation.
When core electrons are treated explicitly or hard pseudopotentials are used, the single-particle basis functions have highly localized features, typically requiring denser grids than are affordable (both in terms of memory and computational cost of the various steps). The solution of the Poisson equation on adaptive grids non-uniformly resolving localized structures is a well-studied problem in computational mathematics. Robust, black-box, highly efficient algorithms are available both to construct adaptive grids which automatically resolve a given density to controllable high-order accuracy, and to solve the Poisson equation on the resulting grids. A primary goal of such solvers is linear or quasi-linear scaling in the number of grid points: that is, the cost of solving the Poisson equation should scale as $\OO{M}$ or $\OO{M \log M}$, where $M$ is the number of points in an adaptive grid discretizing the density.

Many such algorithms are variants of
the fast multipole method (FMM)~\cite{greengard1987thesis,greengard1988,greengard1987jcp,fmm2,gimbutas2003sisc,fmm6,fmm7,ethridge2001sisc,malhotra2016toms} or multigrid
methods~\cite{brandt1977mcom,hackbusch2013,sampath2010sisc,sundar2015nlaa} (see also \cite{biros2016sisc} for a comparison of the FFT, FMM, and multigrid methods, and \cite{pan25} for an example of a multigrid-based approach to compute the ERI tensor). For free space
problems, multigrid-based methods typically require truncating the computational domain and imposing
artificial boundary conditions, so we do not consider this approach.
Here, we use the recently developed
dual-space multilevel kernel-splitting (DMK) framework~\cite{jiang2025cpam}.
DMK shares several characteristics with the FMM, namely its tree-based algorithmic
structure and $\OO{M}$ computational complexity for evaluating the convolution of a kernel
and a function represented on an adaptive grid. However, DMK introduces significant improvements,
including a simplified algorithmic framework, applicability to a broader
class of kernels, and an acceleration of computationally intensive near-field
calculations.

In the present work, we use adaptive octrees with product Chebyshev grids to discretize the basis functions, and the DMK algorithm, which is compatible with such a discretization, to solve the Poisson equation. Once the adaptive octree discretization has been constructed, the DMK solver can be treated as a black box, so we refer to Ref.~\onlinecite{jiang2025cpam} for further details on DMK.

To build the adaptive octree discretization, we begin with the single-particle basis functions $\phi_i(\mathbf{r})$. The tree is built through a recursive subdivision procedure. One begins at the root level of the tree, with a single box covering the full domain. The $\phi_i$ are evaluated on an $n \times n \times n$ product Chebyshev grid in the box, from which one can obtain Chebyshev polynomial interpolants
$p^{(\phi_i,B)}(\mathbf{r}) = \sum_{j,k,l=0}^{n-1} C_{jkl} T_j(\tilde{x}_1)T_k(\tilde{x}_2)T_l(\tilde{x}_3)$. Here $T_j$ is the degree $j$ Chebyshev polynomial of the first kind, $\mathbf{r}=(x_1,x_2,x_3)$, and
$\tilde{x}_i=2(x_i-c_i)/L$ with $c_i$ the center of the box $B$ and $L$ its side length.
The resulting approximation is accurate to order $n+1$, with a fixed, moderate value of $n$ typically chosen.
If the maximum interpolation error over all $\phi_i$ is larger than a user-specified tolerance $\varepsilon$, the box is subdivided into eight equal boxes, and the procedure is repeated for each of the eight new boxes. Here, we use the interpolation error on box $B$ given by
\begin{equation}
    E^{(\phi_i,B)} = \frac{ || \phi_i -  p^{(\phi_i,B)}||_{L^2(B)} }{ || \phi_i ||_{L^2}},
\end{equation}
where the $L^2$ norm in the denominator is taken over the full domain.
This process continues recursively until $E^{(\phi_i)}=\sqrt{\sum_{B} (E^{(\phi_i,B)})^2} < \varepsilon$ for $1 \leq i \leq N$, where the sum is taken over all leaf-level boxes. 
The result of this procedure is a nonuniform collection of the leaf-level boxes, each containing a Chebyshev interpolant. This constitutes a representation of the functions $\phi_i(\mathbf{r})$ accurate to $\varepsilon$ (for instance, it can be evaluated at a point by determining the leaf-level box containing that point, and evaluating the corresponding polynomial interpolant). The boxes will be more refined near localized features of the functions $\phi_i(\mathbf{r})$. It can be shown that the total cost of this adaptive grid construction procedure for a single function is $\OO{M}$, where $M$ is the total number of points in all leaf-level boxes~\cite{jiang2024sirev}. An example of such an octree for single-particle orbitals of $(\mathrm{N}\mathrm{H}_{3})_{2}$ using the aug-cc-pVTZ basis set is shown in Fig.~\ref{fig:octree}.

\begin{figure}[ht]
\centering
\includegraphics[width=0.8\linewidth]{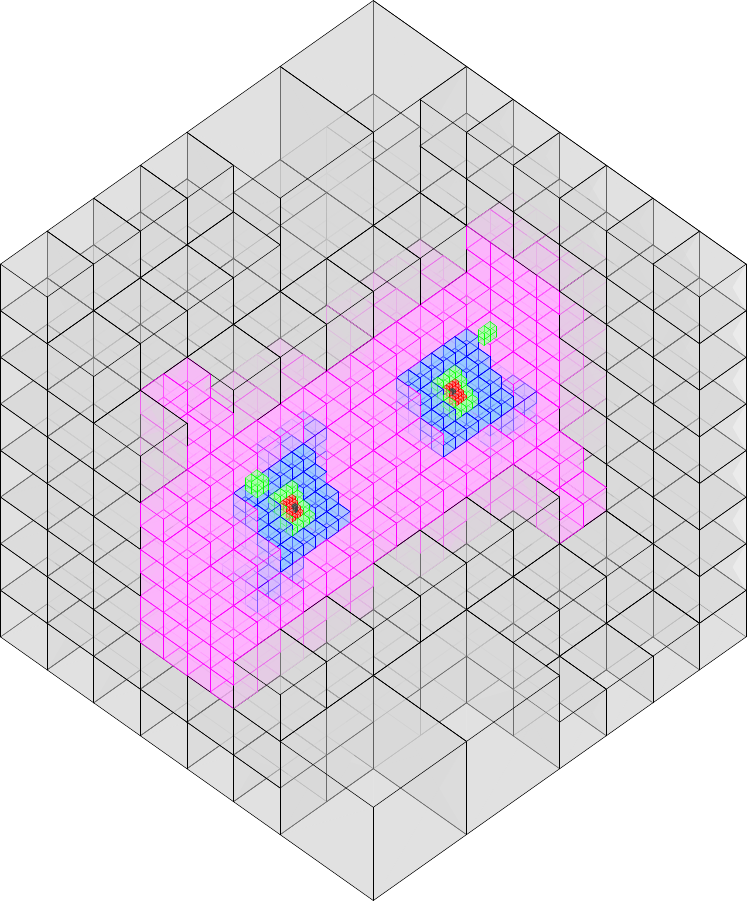} 
\includegraphics[width=0.8\linewidth]{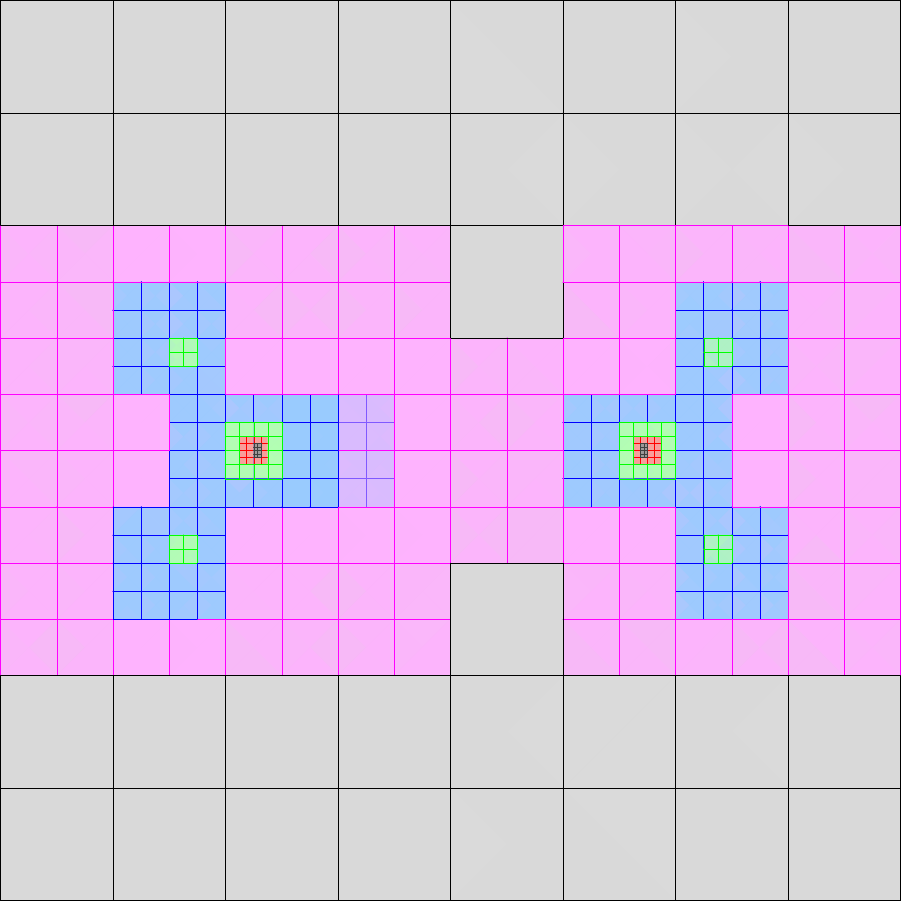} 
\caption{Adaptive octree resolving single-particle orbitals of $(\mathrm{N}\mathrm{H}_{3})_{2}$ using the aug-cc-pVTZ basis set.
  Top: three-dimensional view of the octree showing the adaptively refined boxes near nuclei. Bottom: slice view of the same octree. Colors indicate different refinement levels.}
\label{fig:octree}
\end{figure}

We show in App.~\ref{app:pair_density_error} that if the single-particle basis functions $\phi_i$ are well-approximated on each leaf-level box by a Chebyshev polynomial interpolant of degree $n-1$, the $N^2$ pair densities $\rho_{ij}(\mathbf{r}) = \phi_i(\mathbf{r})\phi_j(\mathbf{r})$ are well-approximated by Chebyshev polynomial interpolants of degree $2n-1$. 
To be more precise, upsampling by increasing the polynomial degree in each box from $n-1$ to $2n-1$ guarantees $E^{(\rho_{ij},B)} \leq C\varepsilon$ for $1 \leq i, j \leq n$ and all leaf boxes $B$, where the constant $C$, given explicitly in App.~\ref{app:pair_density_error}, depends only on $\max_{1 \leq i \leq N} \norm{\phi_i(\mathbf{r})}_{\infty}$
and on $n$, which is fixed. 
We show in App.~\ref{app:isdf} that the ISDF auxiliary basis functions $\zeta_{\mu}(\mathbf{r})$ are linear combinations of the pair densities $\rho_{ij}(\mathbf{r})$, so although this does not prove the same upsampled octree representation is sufficient to resolve them to accuracy $\varepsilon$ (this depends on the properties of the matrix $A$ introduced in App.~\ref{app:isdf}), we find this to be sufficient in practice.

The DMK Poisson solver is directly compatible with the resulting adaptive grid, taking a density on this grid as input and producing the solution of the Poisson equation as a piecewise polynomial interpolant on the same grid. To compute the inner products appearing in \eqref{eq:V_isdf}, we compute the integral on each box using Clenshaw-Curtis quadrature rules \cite{trefethen2008sirev} on the corresponding Chebyshev grids, and then sum the results. 

We emphasize that our adaptive grid construction procedure is agnostic to the
single-particle basis set used, requiring only a basis function evaluator at
arbitrary points. The adaptive grid then serves as an internal representation of
the pair densities used by the DMK Poisson solver. Basis functions can be given
in closed form (e.g., Gaussian-type orbitals or Slater-type orbitals), or
represented in terms of values on a grid which can then be interpolated (e.g., numerical atomic orbitals~\cite{SIESTA_Jose2002,FHIaim_BLUM2009} or adaptive local basis sets~\cite{lin12,kaye15,lin15,wei15,pan25}). 

\subsection{Summary of the algorithm and computational complexity} \label{sec:summary_algorithm} 

We summarize our full procedure as follows.
\begin{enumerate}
    \item Build an adaptive real space grid for the single-particle basis functions $\phi_i(\mathbf{r})$, with error tolerance $\varepsilon$, using the algorithm described in Sec.~\ref{subsec:adapt_poisson}. After upsampling, this discretization is used to represent the pair densities $\rho_{ij}(\mathbf{r})$ and the auxiliary basis functions $\zeta_{\mu}(\mathbf{r})$. The cost of this procedure scales as $\OO{M N}$, since $N$ single-particle basis functions are discretized on a grid of $M$ points. An $\OO{M}$ memory cost can be obtained by building the adaptive grids for each $\phi_i$ sequentially by refining the grid if needed.

    \item Using the ID procedure described in App.~\ref{app:isdf}, select the interpolating points $\mathbf{r}_\mu$ and compute the ISDF auxiliary basis functions $\zeta_\mu(\mathbf{r})$ on the $M$ discretization points. The cost of this procedure scales as $\OO{R^2 M}$, and the memory complexity as $\OO{RM}$, corresponding to the storage of the ISDF auxiliary basis functions on the adaptive grid. 
        
    \item Solve the Poisson equation for each $\zeta_\mu(\mathbf{r})$ on the real space grid using the DMK algorithm. The cost of this step is $\OO{M}$ for each auxiliary basis function, or $\OO{RM}$ in total. Compute the inner products of the results to obtain $V_{\mu\nu}$, at an $\OO{M}$ cost per pair $(\mu, \nu)$, or $\OO{R^2 M}$ in total. Using \eqref{eq:X_isdf} and \eqref{eq:V_isdf}, we can now form the THC decomposition \eqref{eq:eri_thc} of the ERI tensor $V_{ijkl}$. The memory required to store the matrices $X_{i\mu}$ and $V_{\mu\nu}$ in the decomposition scales as $\OO{R^2}$.
\end{enumerate}
We note that $R > N$ since the single-particle basis functions are linearly independent, so the cost of the adaptive grid construction step is asymptotically subdominant. The total computational complexity is therefore $\OO{R^2 M}$, with the ID and inner product steps dominating the cost (and notably not the Poisson solves). Since we expect $R = \OO{N}$, the complexity in terms of $M$ and $N$ is $\OO{M N^2}$, and since furthermore $M = \OO{N}$, the complete algorithm is cubic scaling in system size.

\section{Numerical results \label{sec:results}}

In this section, we demonstrate the accuracy and performance of our combination of ISDF with an adaptive Poisson solver. 
We systematically examine how the error of ISDF-approximated ERIs varies with the adaptive tree tolerance $\varepsilon$ (Sec.~\ref{subsec:tree_tol}), the size $R$ of the ISDF auxiliary basis (Sec.~\ref{subsec:isdf_numerics}), and the basis set size $N$ and locality (Sec.~\ref{sec:basis_size_and_locality}). The improvement over the uniform grid approach is quantified in Sec.~\ref{subsec:adaptibility}.
We then use the approximated ERIs in electronic structure calculations within the $GW$ approximation (Sec.~\ref{subsec:isdf_gw}). 

We use all-electron, atom-centered Gaussian-type orbitals (GTOs) obtained from the \texttt{PySCF} package \cite{PySCF_2020}. 
The basis sets contain Gaussian exponents ranging from approximately $10^{-2}$ to $10^{5}$, with small exponents corresponding to diffuse orbitals and large exponents corresponding to localized core orbitals. This wide range represents an intractable numerical challenge for uniform grid-based methods. 
We emphasize that although we use GTOs for the examples here, our proposed framework is compatible with general basis sets, including non-atom-centered, non-Gaussian, and highly localized functions. 

The ISDF step and subsequent electronic structure evaluations are performed using \texttt{CoQuí}~\cite{CoQuiCode}, a software package for electronic structure simulations beyond density functional theory.
Within \texttt{CoQuí}, the ISDF procedure incorporates adaptive real space grid generation and DMK Poisson solves using the \texttt{dmk} code~\cite{DMKcode}.

\subsection{Accuracy of adaptive grid representation \label{subsec:tree_tol}}

\begin{figure}
\centering
\includegraphics[width=\linewidth]{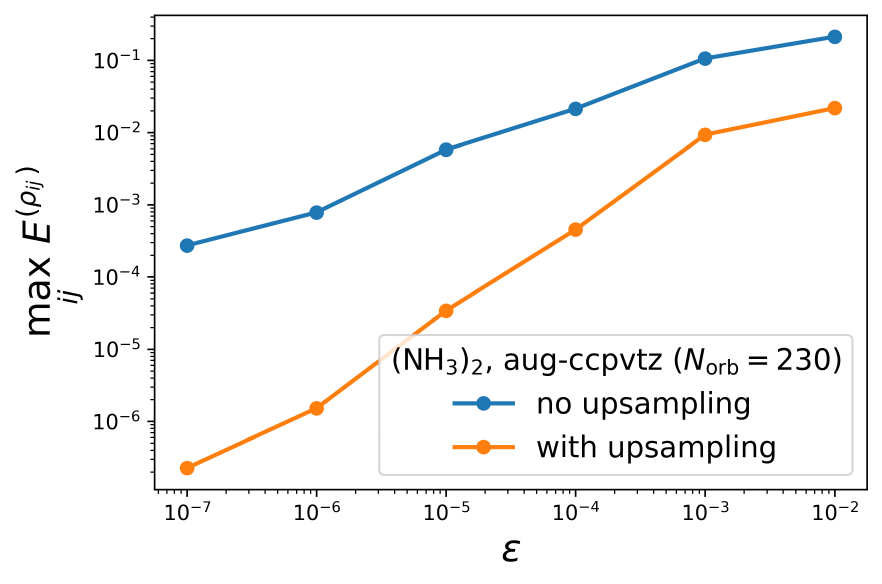} 
\caption{Maximum $L^2$ error of adaptive octree grid representation of all pair densities $\rho_{ij}$, versus adaptive grid tolerance, for the $(\mathrm{N}\mathrm{H}_{3})_{2}$ example. We show the error of the interpolant on the original adaptive grid used to resolve the single-particle orbitals (blue, no upsampling), as well as for a grid upsampled by a factor of $1.5$ in each dimension (orange, with upsampling). We use the polynomial order $n = \log_{10}(\varepsilon^{-1})+1$ in the octree grid construction.}
\label{fig:adapt_resample}
\end{figure}

We first demonstrate our assertion, described in Sec.~\ref{subsec:adapt_poisson} and App.~\ref{app:pair_density_error}, that upsampling the adaptive grid used to represent the single-particle orbitals $\phi_{i}(\mathbf{r})$ effectively resolves the pair densities $\rho_{ij}(\mathbf{r})$. Though it is shown in App.~\ref{app:pair_density_error} that an upsampling factor of $2$ leads to a rigorous error bound, we find in practice that a factor $1.5$ is sufficient, and use this value. 
We take the polynomial order $n = \log_{10}(\varepsilon^{-1})+1$ in the adaptive grid construction, so that higher-order polynomials are used for more stringent tolerances.

We consider the ammonia dimer $(\mathrm{N}\mathrm{H}_{3})_{2}$ with aug-cc-pVTZ basis set. 
Fig.~\ref{fig:adapt_resample} shows the maximum interpolation error of any $\rho_{ij}$ using the original adaptive grid resolving the $\phi_i$, and the upsampled grid, varying the adaptive grid generation tolerance $\varepsilon$. We observe that the tolerance is achieved within one digit for all $\rho_{ij}$ on the upsampled grid.

\subsection{Accuracy of ISDF truncation \label{subsec:isdf_numerics}}

\begin{figure}
\centering
\includegraphics[width=\linewidth]{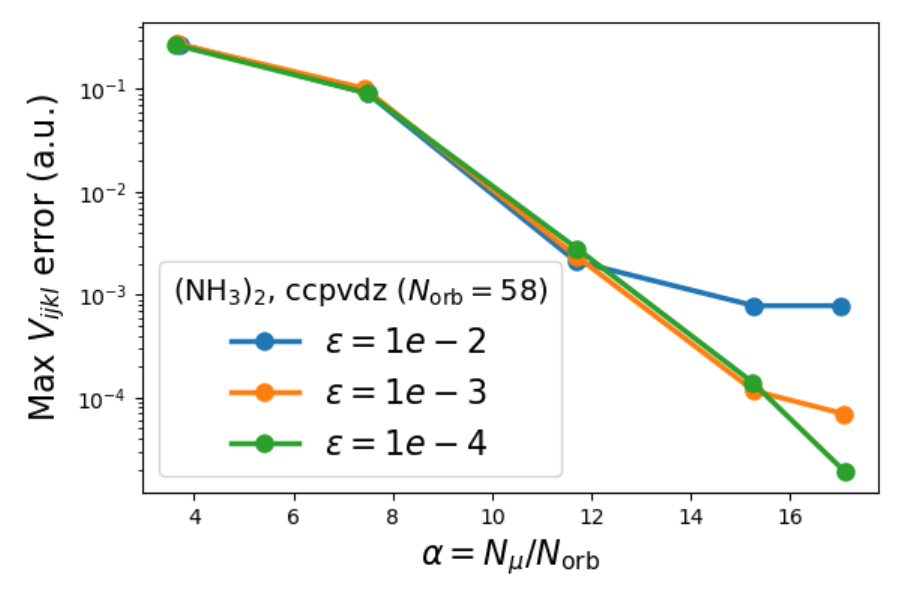} 
\caption{Maximum error of ERI tensor $V_{ijkl}$ for $(\mathrm{NH}_{3})_{2}$, varying the ISDF truncation rank $R$ of the pair densities and the adaptive grid tolerance $\varepsilon$.}
\label{fig:isdf_octree}
\end{figure}

We next examine the convergence of the ISDF-approximated ERI tensor $V_{ijkl}$ with the ISDF rank $R$. Since we expect the numerical rank of the pair density matrix $\rho_{ij}(\mathbf{r})$ to scale as $R = \mathcal{O}(N)$, we define the proportionality factor $\alpha = R/N$, the ratio of the number of auxiliary basis functions to the number of single-particle basis functions. In practice, the ISDF error is controlled by either varying the ID rank $R$ or the ID error tolerance (using a rank-revealing algorithm), but we report $\alpha$ as a measure of the efficiency of the ISDF approximation in compressing the pair densities.

In Fig.~\ref{fig:isdf_octree}, for the same system as above,  we plot the maximum error in $V_{ijkl}$ as a function of $\alpha$ for different values of the adaptive grid tolerance $\varepsilon$. Reference values of $V_{ijkl}$ are computed analytically using \texttt{PySCF}. 
We observe exponential convergence of the ISDF error with respect to $\alpha$ (equivalently, $R$, since $N$ is fixed), plateauing at a value determined by $\varepsilon$. 
This indicates that $\varepsilon$ controls the accuracy of the THC-ERI approximation, and we recommend selecting it as such in practice.

\subsection{Increasing basis set size and locality\label{sec:basis_size_and_locality}}

\begin{figure}
\centering
  \includegraphics[width=\linewidth]{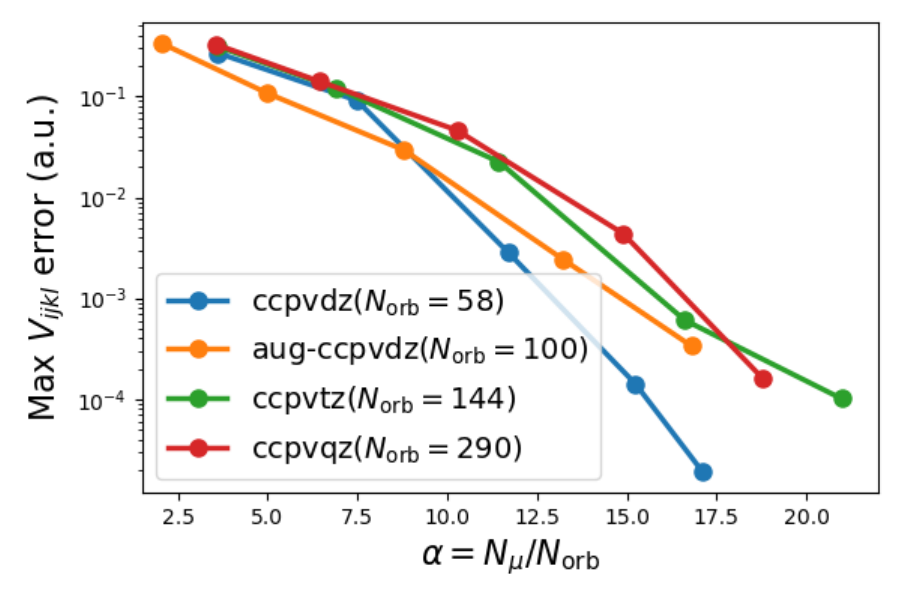} 
\caption{Maximum error of the ERI tensor $V_{ijkl}$ for $(\mathrm{NH}_{3})_{2}$, varying the ISDF truncation rank of the pair densities, using all-electron cc-PVXZ basis sets with X = D, T, and Q.}
\label{fig:basis_conv}
\end{figure}

We next investigate the performance of the algorithm with larger and more localized basis sets, in each case measuring the maximum error of the ERI tensor against $\alpha$, as in Fig.~\ref{fig:isdf_octree}.
In Fig.~\ref{fig:basis_conv}, we again examine the $(\mathrm{NH}_{3})_{2}$ molecule, now using the cc-pVXZ series of all-electron basis sets with X = D, T, Q. The adaptive grid tolerance is fixed at $\varepsilon=10^{-4}$.
For all three basis sets, the ISDF approximation systematically improves with increasing $\alpha$, and we observe overall insensitivity of the convergence behavior and compression factor $\alpha$ to the basis set size. 
This controlled convergence contrasts with resolution-of-identity (RI) or density fitting (DF) methods, in which the accuracy is inherently limited by that of pre-optimized auxiliary basis sets tailored to each single-particle basis, rather than the systematically computed ISDF auxiliary basis~\cite{DF_Ren2012}. 

To achieve chemical accuracy across all ERI elements, we find that $\alpha \approx 16$ is sufficient for all three basis sets. 
This value is larger than typical values used in periodic systems with pseudopotentials treated with uniform grids, for which $\alpha\approx 8$ often suffices~\cite{THC-RPA_CNY2023}. 
For certain systems, we empirically observe a weak dependence of $\alpha$ on the required grid resolution, but we have not systematically studied this scaling.
We emphasize that the larger values of $\alpha$ are not directly a consequence of the use of adaptive grids, but rather the locality of all-electron basis functions compared to those used in pseudopotential calculations, for which Gaussians exponents are typically approximately $10$ at most. 
Furthermore, demanding chemical accuracy in each element of $V_{ijkl}$ is typically an unnecessarily stringent criterion, since one is usually interested in observables such as the total energy and orbital energies. 
We lastly note that these larger values of $\alpha$ are required primarily to resolve sharp features associated with core electrons, which mostly contribute to the Hartree potential. 
In Sec.~\ref{subsec:isdf_gw}, we describe a hybrid method for electronic structure calculations which allows for smaller $\alpha$ by observing that the Hartree potential term can be handled efficiently without the use of THC-ERIs.

\begin{figure}
\centering
\includegraphics[width=\linewidth]{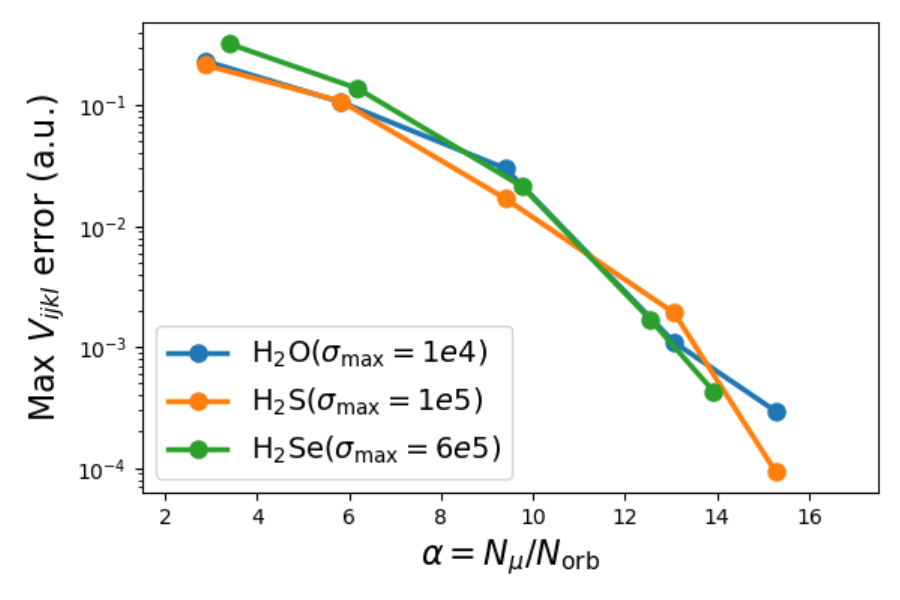} 
\caption{Maximum error of the ERI tensor $V_{ijkl}$ for chalcogen hydrides with increasing atomic numbers, varying the ISDF truncation rank of the pair densities. The largest GTO exponents $\sigma_{\mathrm{max}}$ for each system are shown in parentheses.}
\label{fig:chalcogen_hydrides}
\end{figure}

We next study a sequence of systems for which the single-particle basis functions become increasingly localized. Fig.~\ref{fig:chalcogen_hydrides} shows the maximum error of the ERI tensor as a function of $\alpha$ for a series of chalcogen hydrides---$\mathrm{H_{2}O}$, $\mathrm{H_{2}S}$, and $\mathrm{H_{2}Se}$---using the aug-ccPVDZ basis set. As the chalcogen atom becomes heavier, the largest Gaussian exponent in the basis set increases by roughly an order of magnitude, leading to significantly more localized orbitals. 
We observe consistent convergence behavior across the entire series, and in this case $\alpha$ does not increase with basis set localization,
even though the degree of spatial localization in these all-electron GTO basis sets is orders of magnitude stronger than in basis sets used with pseudopotentials or effective core potentials. An FFT-based ISDF implementation is impractical for all three of these systems, due to the large spatial resolution required.

\begin{figure}
\centering
\includegraphics[width=\linewidth]{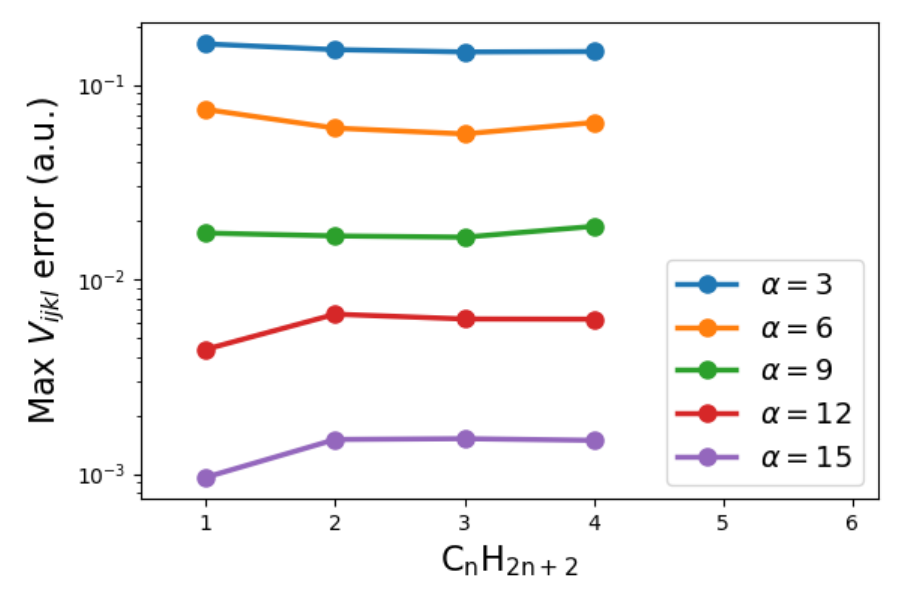} 
\caption{Maximum error of the ERI tensor $V_{ijkl}$ for alkanes $\mathrm{C}_{n}\mathrm{H}_{2n+2}$ varying $n$ and the compression factor $\alpha$.}
\label{fig:alkanes}
\end{figure}

We next investigate performance as the number of atoms in a molecule is increased, using the alkane series $\mathrm{C}_{n}\mathrm{H}_{2n+2}$, for which increasing $n$ systematically increases the size of the molecule and the number of orbitals in the system. We use the aug-cc-pVDZ basis set.
In Fig.~\ref{fig:alkanes}, we plot the maximum error of the ERI tensor for alkanes ranging from methane ($n=1$) to pentane ($n=5$), for different values of $\alpha$. As $n$ increases, the ISDF error at a fixed $\alpha$ remains roughly constant, demonstrating the scaling $R = \OO{N}$.

\subsection{Comparison with uniform grid approach, and timings \label{subsec:adaptibility}}
\begin{table*}[bth]
\centering
\begin{tabular}{c|c|c|c|c}
\hline
\hline
System & $\varepsilon$ & Setting & Adaptive octree DOF & Uniform grid DOF \\
\hline
& 
$10^{-3}$ & ccECP           ($\sigma_{\max} \approx 42$) & $775872$ & Between $128^3 \approx 2.1 \times 10^6$ and $256^3 \approx 1.7 \times 10^7$ \\
$(\mathrm{NH}_3)_2$ & 		& all-electron ($\sigma_{\max} \approx 1 \times 10^4$) & $854496$ & Between $256^3 \approx 1.7 \times 10^7$ and $512^3 \approx 1.3 \times 10^8$  \\
\cline{2-5}
&  $10^{-5}$ & ccECP           ($\sigma_{\max} \approx 42$) & $3577932$ & Between $256^3 \approx 1.7 \times 10^7$ and $512^3 \approx 1.3 \times 10^8$ \\
& 		& all-electron ($\sigma_{\max} \approx 1 \times 10^4$) & $4353588$ & Between $512^3 \approx 1.3 \times 10^8$ and $1024^3 \approx 1.1 \times 10^{9}$ \\
\hline
& 
$10^{-3}$ & ccECP           ($\sigma_{\max} \approx 85$) & $588384$ & Between $256^3 \approx 1.7 \times 10^7$ and $512^3 \approx 1.3 \times 10^8$ \\
$\mathrm{TiO}$ & 		& all-electron ($\sigma_{\max} \approx 3 \times 10^6$) & $715392$ & Between $16384^3 \approx 4.4 \times 10^{12}$ and $32768^{3} \approx 3.5 \times 10^{13}$  \\
\cline{2-5}
&  $10^{-5}$ & ccECP           ($\sigma_{\max} \approx 85$) & $2781864$ & Between $512^3 \approx 1.3 \times 10^8$ and $1024^3 \approx 1.1 \times 10^{9}$ \\
& 		& all-electron ($\sigma_{\max} \approx 3 \times 10^6$) & $3414636$ & Between $65536^3 \approx 2.8 \times 10^{14}$ and $131072^{3} \approx 2.3 \times 10^{15}$ \\
\hline
\hline
\end{tabular}
\caption{Number of adaptive real space grid points and estimated number of uniform grid points required to resolve effective core potential (ECP) and all-electron basis sets for $(\mathrm{NH}_3)_2$ and $\mathrm{TiO}$ at different target error tolerances $\varepsilon$.}
\label{tab:adaptive_vs_uniform}
\end{table*}

In this section, we compare the sizes of the real space grids obtained using our adaptive grid construction with those of a conventional uniform grid.
Since all computational bottleneck steps of the ISDF procedure scale linearly (or quasi-linearly, if FFT is used as a Poisson solver) with the grid size $M$, the ratio of grid sizes provides a reasonable proxy for the speedup achieved by our approach over an FFT-based ISDF implementation. We also show example wall clock timings for the various steps of our procedure.

For the analysis of grid sizes, we consider TiO and $(\mathrm{NH}_{3})_{2}$ molecules in two numerical setups: (i) a cc-PVTZ-ccECP basis set and an effective core potential (ECP)~\cite{ccecp_2ndrow_2017,ccecp_tm_2018}, and (ii) an all-electron cc-PVTZ basis set~\cite{ccpvtz_3d_2005}. 
For the uniform grid, the number of grid points per dimension required to achieve a target accuracy $\varepsilon$ is estimated by resolving each basis function along a representative one-dimensional slice, since resolving the localized basis functions with a uniform three-dimensional grid is intractable. In practice, we estimate the number of grid points required per dimension by doubling the grid size until the error falls below the desired tolerance, so we report the corresponding lower and upper bound estimates for each target tolerance.

Table~\ref{tab:adaptive_vs_uniform} compares the required grid sizes for target relative errors $\varepsilon = 10^{-3}$ and $10^{-5}$. 
As expected, the ratio between uniform and adaptive grid sizes is several orders of magnitude larger in the all-electron setting compared to the ECP case, especially for $\mathrm{TiO}$, since the largest Gaussian exponent increases from approximately $10^{2}$ to $10^{6}$.
This is because the uniform grid size increases in proportion to the largest Gaussian exponent, whereas the adaptive grid size increases only logarithmically. Still, even in the ECP case, the adaptive grid is significantly more compact than the uniform grid.

\begin{table}[bth]
\begin{tabular}{c|c|c|c|c}
\hline
\hline
C$_{n}$H$_{2n+2}$ & $n=1$ & $n=2$ & $n=3$ & $n=4$ \\ 
\hline
Octree construction& 0.09 & 0.21 & 0.38 & 0.55 \\
ISDF generation of $\mathrm{r}_{\mu}$ and $\zeta_{\mu}(\mathbf{r})$& 4.86 & 16.05 & 39.97 & 80.18\\
Poisson solves for $\zeta_{\mu}(\mathbf{r})$ & 7.77 & 18.69 & 31.35 & 53.12\\
\hline
\hline
\end{tabular}
\caption{Wall clock timings in minutes for major steps of the adaptive ISDF procedure applied to C$_{n}$H$_{2n+2}$ with the all-electron aug-cc-PVDZ basis set. The octree tolerance is set to $\varepsilon = 10^{-5}$, and the ISDF auxiliary basis size is fixed to $\alpha=12$. All calculations are performed using 8 OpenMP threads.}
\label{tab:timing}
\end{table}

Table~\ref{tab:timing} summarizes timings of the main steps of our workflow for the alkanes C$_{n}$H$_{2n+2}$ using the all-electron aug-cc-PVDZ basis set and an octree tolerance $\varepsilon = 10^{-5}$. The ISDF auxiliary basis size is fixed to $\alpha=12$, and all calculations are performed using 8 OpenMP threads. 

The cost of the octree grid construction for the $\phi_i$ and $\rho_{ij}$ scales only quadratically with the system size and is negligible. 
As the system size increases, the cost of the adaptive Poisson solver step (less than one second per solve) becomes smaller than that of the ISDF step. Indeed, the Poisson solver step also scales only quadratically with system size and is trivially parallelizable over the auxiliary basis function index, whereas the ISDF step scales cubically. Thus, for systems with a large number of orbitals, the cost of the adaptive Poisson solver will become negligible relative to that of the ISDF step. 

Since the cost of the dominant ISDF step depends only on the number of grid points, and not on their spatial arrangement, we expect that for sufficiently large systems the cost \textit{per grid point} of our fully adaptive ISDF framework will be similar to that of the standard FFT-based uniform grid approach, i.e., the overhead associated with grid adaptivity will become negligible. A caveat is that the required compression factor $\alpha$ might be larger for systems requiring more grid adaptivity, but our numerical experiments suggest at most mild growth (see Figs.~\ref{fig:basis_conv} and \ref{fig:chalcogen_hydrides}). We note that the timing of the inner product in \eqref{eq:V_isdf} is not shown, since it is small compared to that of the ISDF step and has an identical cost scaling. 

\subsection{Correlated electron structure using \texorpdfstring{$GW$}{GW} \label{subsec:isdf_gw}}
\begin{figure}
\centering
\includegraphics[width=\linewidth]{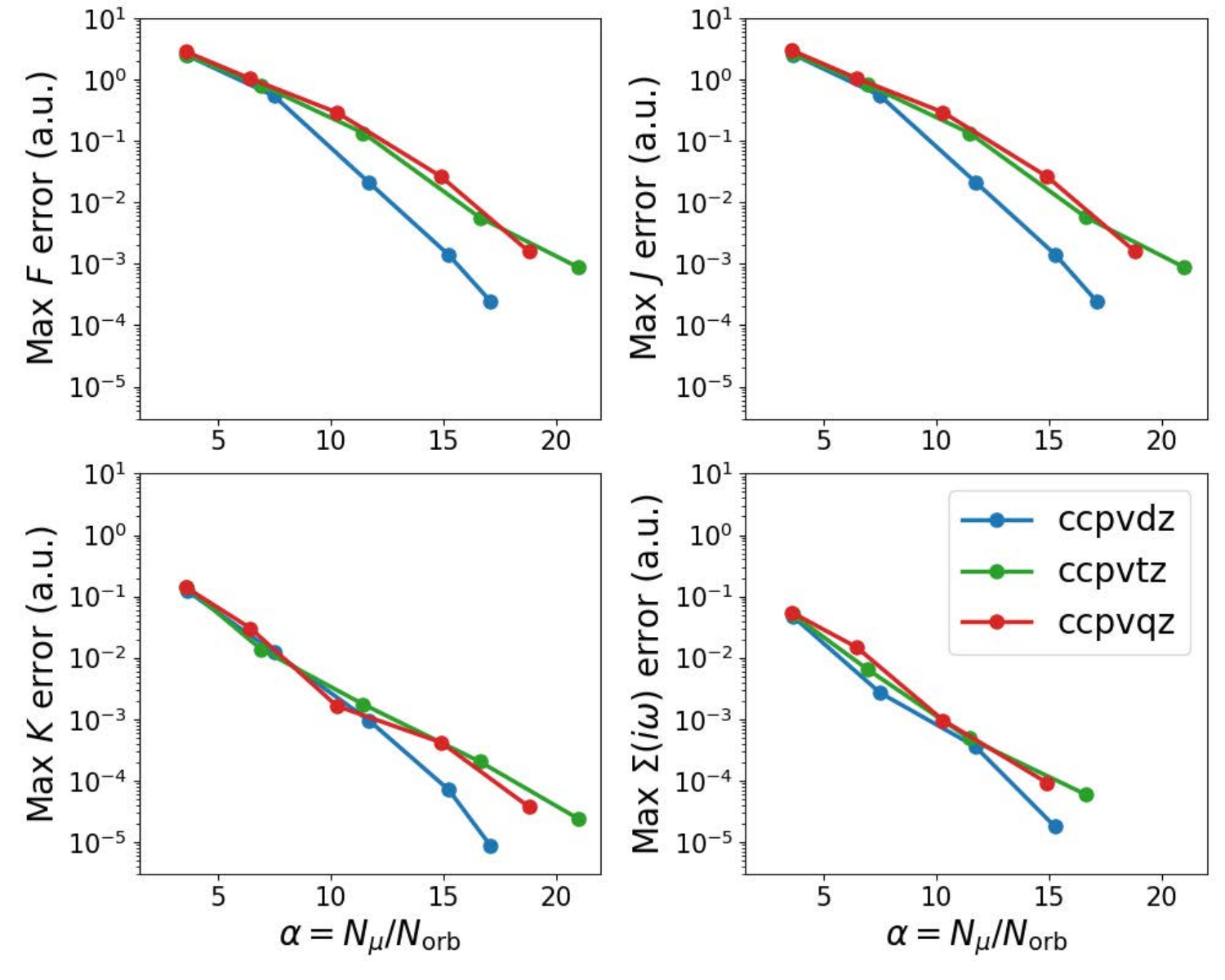} 
\caption{Convergence of the Fock matrix (top left) with the compression factor $\alpha$, including Hartree (top right) and exchange (bottom left) contributions, for $(\mathrm{NH}_{3})_{2}$ example, as well as the dynamic $GW$ self-energy (bottom right).}
\label{fig:isdf_gw}
\end{figure}

\begin{figure}
\centering
\includegraphics[width=\linewidth]{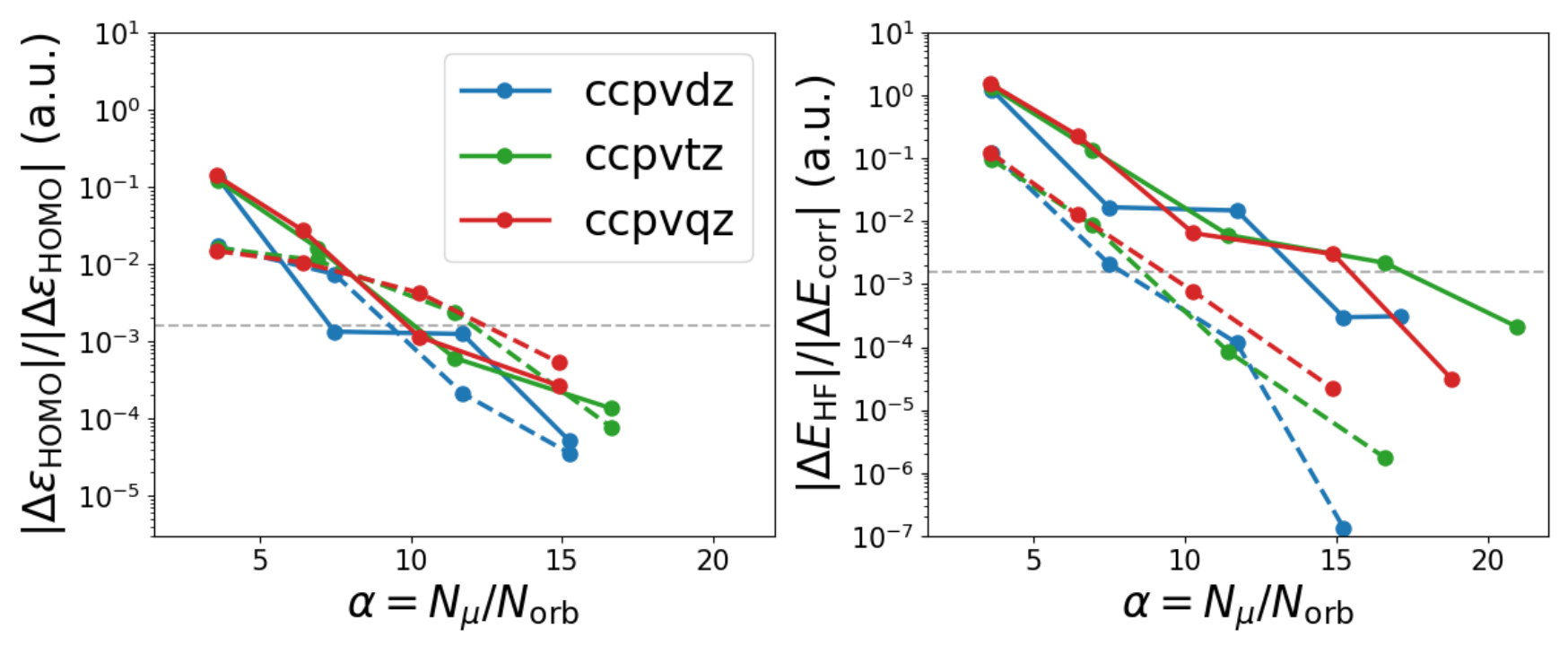} 
\caption{Convergence of orbital energies and total energies with respect to the compression factor $\alpha$ for $(\mathrm{NH}_{3})_{2}$. 
The left panel shows the HOMO (solid lines) and LUMO (dashed lines) energies, while the right panel reports the correlation energy (solid lines) and the Hartree-Fock energy (dashed lines).}
\label{fig:isdf_gw_energies}
\end{figure}

We next assess the performance of our framework in correlated electronic structure calculations within the $GW$ approximation, using the THC decomposition of the ERI tensor introduced in Eqs.~\ref{eq:thc_K} and~\ref{eq:thc_GW}. Implementation details of the THC–$GW$ algorithm are described in Ref.~\onlinecite{Yeh_ISDF_GW2024}.

The $GW$ self-energy consists of two components: the static Hartree-Fock potential (the Fock matrix $F$), which provides the mean-field description of electrons, and the dynamic self-energy $\Sigma(i\omega)$, a function of Matsubara frequencies $i\omega$, which captures dynamic electron correlations. 
As shown at the beginning of Sec.~\ref{sec:THC}, both terms depend explicitly on the ERI tensor, and are therefore directly influenced by the accuracy of the ISDF approximation. 

We first examine the static mean-field contribution using the $(\mathrm{NH}_3)_2$ dimer as an example. Fig.~\ref{fig:isdf_gw} plots the maximum error in the Fock matrix $F$ as a function of the compression factor $\alpha$, together with the corresponding errors in its Hartree ($J$) and exchange ($K$) components. 

The ISDF errors in the Fock matrix $F$ are noticeably larger than those in the ERI tensor $V$ reported in Fig.~\ref{fig:basis_conv}.  
This is consistent with prior observations in THC-based electronic structure methods, for which more auxiliary functions are needed to reach satisfactory HF accuracy, limiting the use of THC for HF calculations~\cite{THC_MP3_Joonho2020,THC_GWSOX_Pokhilko2024}. 
More specifically, our data suggest that the dominant contribution to the Fock matrix error is the Hartree term, whereas the exchange potential shows convergence comparable to or better than that of the ERI tensor itself. 
This is likely a consequence of the larger overall magnitude of the Hartree potential. 
This result motivates a practical hybrid scheme for HF calculations: compute the Hartree potential directly using the adaptive grid DMK method and use ISDF only for the exchange term. The total cost of the former scales as $\OO{M N^2}$, less than the $\OO{M R^2}$ cost of our ISDF scheme. 
Within this hybrid approach, one could potentially reduce the required ISDF compression factor $\alpha$ even in an HF calculation. 

Fig.~\ref{fig:isdf_gw} also shows the maximum error of the dynamic $GW$ self-energy in Eq.~\ref{eq:thc_GW} across all orbitals and Matsubara frequencies. 
We observe convergence similar to that for the exchange potential.
Importantly, computing the $GW$ self-energy involves contractions over the full $V_{ijkl}$ tensor, including the occupied-virtual and virtual-virtual blocks. In contrast, the Hartree-Fock potential only depends on the occupied subspace. The fact that both the dynamic self-energy and the exchange potential exhibit consistent convergence is a consequence of the accurate representation of the full ERI tensor by the ISDF approach, confirming the robustness and transferability of the ISDF approximation beyond mean-field theory.

To further assess the impact of ISDF errors on physical observables, Fig.~\ref{fig:isdf_gw_energies} shows the convergence of orbital energies (left panel) and total energies (right panel) with respect to $\alpha$. Here the hybrid approach is adopted in which only the exchange and dynamic self-energy terms are evaluated with THC-ERIs. 

The convergence of both the highest occupied molecular orbital (HOMO) and lowest unoccupied molecular orbital (LUMO) energies closely follow that of the exchange potential and dynamic self-energy, with chemical accuracy reached at $\alpha \approx 12$. 
Similar behavior is observed for $GW$ correlation eneries. 
Notably, the errors in HOMO/LUMO energies and correlation energies are consistently much smaller than the errors of the individual ERI elements (see Fig.~\ref{fig:basis_conv}). 
This suggests that the ERI components most relevant for correlation effects, which are dominated by valence orbitals, are well represented within the all-electron THC-ERI framework. 
By contrast, the HF energy contains large contributions from core orbitals, which are more sensitive to the THC-ERI errors, and therefore converges more slowly. 
While the hybrid approach significantly improves the accuracy of the HF energies, convergence still requires $\alpha \approx 16$.
These results indicate that for total energy calculations requiring high precision, ISDF-compressed ERIs are most efficient for post-HF treatment. 

\section{Conclusion \label{sec:conclusion}}
We have developed a cubic-scaling framework for constructing a THC representation of the ERIs for arbitrary smooth single-particle basis functions. 
The strength of our approach stems from the complementary advantages of ISDF and adaptive grid techniques. 
ISDF enables the construction of compact, on-the-fly auxiliary bases, circumventing the need for pre-optimized auxiliary basis functions tailored to specific single-particle basis sets. ISDF also scales cubically with system size, in contrast to alternatives like least squares THC, which scale quartically. Using an adaptive Poisson solver allows for fast and accurate evaluation of Coulomb matrix elements in the ISDF auxiliary basis, with a computational cost scaling linearly with the number of real space grid points even for highly localized basis functions. 
A key component of the method is the adaptive piecewise polynomial representation of the ISDF auxiliary basis functions $\zeta_\mu(\mathbf{r})$, which enables ISDF to handle general basis functions of arbitrary shape. Since the total cost of the adaptive Poisson solves scales only quadratically with the system size and is mild in practice, and the ISDF compression ratio $\alpha$ appears to be at most mildly larger for systems with highly localized orbitals, our results show that ISDF can be performed with a given number of grid points arranged uniformly or adaptively with a comparable cost. 

Our approach addresses the primary computational bottleneck in THC-HF and THC-$GW$ 
calculations involving sharply localized basis functions, yielding a fully 
cubic-scaling workflow that makes large-scale all-electron calculations practical. 
The explicit inclusion of core orbitals is necessary for describing core-level 
excitations and for accurately capturing many-body renormalization effects between 
core and valence bands. 

An important direction for our future work is extending our framework to periodic systems, for which the simultaneous separation of $k$ point 
and orbital indices makes ISDF even more advantageous.
This will involve generalizing the DMK algorithm to periodic boundary conditions 
for $\Gamma$-point calculations, as described in 
Ref.~\onlinecite{dmk_periodic_Ludvig2025}, and quasi-periodic 
boundary conditions for $k$-point sampling.

\begin{acknowledgments}

The Flatiron Institute is a division of the Simons Foundation.

\end{acknowledgments}

\bibliography{refs2,fmm}

\begin{thebibliography}{81}%
\makeatletter
\providecommand \@ifxundefined [1]{%
 \@ifx{#1\undefined}
}%
\providecommand \@ifnum [1]{%
 \ifnum #1\expandafter \@firstoftwo
 \else \expandafter \@secondoftwo
 \fi
}%
\providecommand \@ifx [1]{%
 \ifx #1\expandafter \@firstoftwo
 \else \expandafter \@secondoftwo
 \fi
}%
\providecommand \natexlab [1]{#1}%
\providecommand \enquote  [1]{``#1''}%
\providecommand \bibnamefont  [1]{#1}%
\providecommand \bibfnamefont [1]{#1}%
\providecommand \citenamefont [1]{#1}%
\providecommand \href@noop [0]{\@secondoftwo}%
\providecommand \href [0]{\begingroup \@sanitize@url \@href}%
\providecommand \@href[1]{\@@startlink{#1}\@@href}%
\providecommand \@@href[1]{\endgroup#1\@@endlink}%
\providecommand \@sanitize@url [0]{\catcode `\\12\catcode `\$12\catcode
  `\&12\catcode `\#12\catcode `\^12\catcode `\_12\catcode `\%12\relax}%
\providecommand \@@startlink[1]{}%
\providecommand \@@endlink[0]{}%
\providecommand \url  [0]{\begingroup\@sanitize@url \@url }%
\providecommand \@url [1]{\endgroup\@href {#1}{\urlprefix }}%
\providecommand \urlprefix  [0]{URL }%
\providecommand \Eprint [0]{\href }%
\providecommand \doibase [0]{https://doi.org/}%
\providecommand \selectlanguage [0]{\@gobble}%
\providecommand \bibinfo  [0]{\@secondoftwo}%
\providecommand \bibfield  [0]{\@secondoftwo}%
\providecommand \translation [1]{[#1]}%
\providecommand \BibitemOpen [0]{}%
\providecommand \bibitemStop [0]{}%
\providecommand \bibitemNoStop [0]{.\EOS\space}%
\providecommand \EOS [0]{\spacefactor3000\relax}%
\providecommand \BibitemShut  [1]{\csname bibitem#1\endcsname}%
\let\auto@bib@innerbib\@empty
\bibitem [{\citenamefont {Werner}\ \emph {et~al.}(2003)\citenamefont {Werner},
  \citenamefont {Manby},\ and\ \citenamefont {Knowles}}]{DF_Werner2003}%
  \BibitemOpen
  \bibfield  {author} {\bibinfo {author} {\bibfnamefont {H.-J.}\ \bibnamefont
  {Werner}}, \bibinfo {author} {\bibfnamefont {F.~R.}\ \bibnamefont {Manby}},\
  and\ \bibinfo {author} {\bibfnamefont {P.~J.}\ \bibnamefont {Knowles}},\
  }\bibfield  {title} {\bibinfo {title} {Fast linear scaling second-order
  {M{\o}ller-Plesset} perturbation theory ({MP}2) using local and density
  fitting approximations},\ }\href {https://doi.org/10.1063/1.1564816}
  {\bibfield  {journal} {\bibinfo  {journal} {J. Chem. Phys.}\ }\textbf
  {\bibinfo {volume} {118}},\ \bibinfo {pages} {8149} (\bibinfo {year}
  {2003})}\BibitemShut {NoStop}%
\bibitem [{\citenamefont {Weigend}\ \emph {et~al.}(2009)\citenamefont
  {Weigend}, \citenamefont {Kattannek},\ and\ \citenamefont
  {Ahlrichs}}]{CD_DF_Weigend2009}%
  \BibitemOpen
  \bibfield  {author} {\bibinfo {author} {\bibfnamefont {F.}~\bibnamefont
  {Weigend}}, \bibinfo {author} {\bibfnamefont {M.}~\bibnamefont {Kattannek}},\
  and\ \bibinfo {author} {\bibfnamefont {R.}~\bibnamefont {Ahlrichs}},\
  }\bibfield  {title} {\bibinfo {title} {Approximated electron repulsion
  integrals: Cholesky decomposition versus resolution of the identity
  methods},\ }\href {https://doi.org/10.1063/1.3116103} {\bibfield  {journal}
  {\bibinfo  {journal} {J. Chem. Phys.}\ }\textbf {\bibinfo {volume} {130}},\
  \bibinfo {pages} {164106} (\bibinfo {year} {2009})}\BibitemShut {NoStop}%
\bibitem [{\citenamefont {Ren}\ \emph {et~al.}(2012)\citenamefont {Ren},
  \citenamefont {Rinke}, \citenamefont {Blum}, \citenamefont {Wieferink},
  \citenamefont {Tkatchenko}, \citenamefont {Sanfilippo}, \citenamefont
  {Reuter},\ and\ \citenamefont {Scheffler}}]{DF_Ren2012}%
  \BibitemOpen
  \bibfield  {author} {\bibinfo {author} {\bibfnamefont {X.}~\bibnamefont
  {Ren}}, \bibinfo {author} {\bibfnamefont {P.}~\bibnamefont {Rinke}}, \bibinfo
  {author} {\bibfnamefont {V.}~\bibnamefont {Blum}}, \bibinfo {author}
  {\bibfnamefont {J.}~\bibnamefont {Wieferink}}, \bibinfo {author}
  {\bibfnamefont {A.}~\bibnamefont {Tkatchenko}}, \bibinfo {author}
  {\bibfnamefont {A.}~\bibnamefont {Sanfilippo}}, \bibinfo {author}
  {\bibfnamefont {K.}~\bibnamefont {Reuter}},\ and\ \bibinfo {author}
  {\bibfnamefont {M.}~\bibnamefont {Scheffler}},\ }\bibfield  {title} {\bibinfo
  {title} {Resolution-of-identity approach to {H}artree{\textendash}{F}ock,
  hybrid density functionals, {RPA}, {MP}2 and {$GW$} with numeric
  atom-centered orbital basis functions},\ }\href
  {https://doi.org/10.1088/1367-2630/14/5/053020} {\bibfield  {journal}
  {\bibinfo  {journal} {New J. Phys.}\ }\textbf {\bibinfo {volume} {14}},\
  \bibinfo {pages} {053020} (\bibinfo {year} {2012})}\BibitemShut {NoStop}%
\bibitem [{\citenamefont {Sun}\ \emph {et~al.}(2017)\citenamefont {Sun},
  \citenamefont {Berkelbach}, \citenamefont {McClain},\ and\ \citenamefont
  {Chan}}]{GDF_MDF_Sun2017}%
  \BibitemOpen
  \bibfield  {author} {\bibinfo {author} {\bibfnamefont {Q.}~\bibnamefont
  {Sun}}, \bibinfo {author} {\bibfnamefont {T.~C.}\ \bibnamefont {Berkelbach}},
  \bibinfo {author} {\bibfnamefont {J.~D.}\ \bibnamefont {McClain}},\ and\
  \bibinfo {author} {\bibfnamefont {G.~K.-L.}\ \bibnamefont {Chan}},\
  }\bibfield  {title} {\bibinfo {title} {Gaussian and plane-wave mixed density
  fitting for periodic systems},\ }\href {https://doi.org/10.1063/1.4998644}
  {\bibfield  {journal} {\bibinfo  {journal} {J. Chem. Phys.}\ }\textbf
  {\bibinfo {volume} {147}},\ \bibinfo {pages} {164119} (\bibinfo {year}
  {2017})}\BibitemShut {NoStop}%
\bibitem [{\citenamefont {Ye}\ and\ \citenamefont
  {Berkelbach}(2021)}]{RSDF_HongZhou2021}%
  \BibitemOpen
  \bibfield  {author} {\bibinfo {author} {\bibfnamefont {H.-Z.}\ \bibnamefont
  {Ye}}\ and\ \bibinfo {author} {\bibfnamefont {T.~C.}\ \bibnamefont
  {Berkelbach}},\ }\bibfield  {title} {\bibinfo {title} {Fast periodic
  {G}aussian density fitting by range separation},\ }\href
  {https://doi.org/10.1063/5.0046617} {\bibfield  {journal} {\bibinfo
  {journal} {J. Chem. Phys.}\ }\textbf {\bibinfo {volume} {154}},\ \bibinfo
  {pages} {131104} (\bibinfo {year} {2021})}\BibitemShut {NoStop}%
\bibitem [{\citenamefont {Beebe}\ and\ \citenamefont
  {Linderberg}(1977)}]{Cholesky_ERI_MOL_Beebe1977}%
  \BibitemOpen
  \bibfield  {author} {\bibinfo {author} {\bibfnamefont {N.~H.~F.}\
  \bibnamefont {Beebe}}\ and\ \bibinfo {author} {\bibfnamefont
  {J.}~\bibnamefont {Linderberg}},\ }\bibfield  {title} {\bibinfo {title}
  {Simplifications in the generation and transformation of two-electron
  integrals in molecular calculations},\ }\href
  {https://doi.org/https://doi.org/10.1002/qua.560120408} {\bibfield  {journal}
  {\bibinfo  {journal} {Int. J. Quantum Chem.}\ }\textbf {\bibinfo {volume}
  {12}},\ \bibinfo {pages} {683} (\bibinfo {year} {1977})}\BibitemShut
  {NoStop}%
\bibitem [{\citenamefont {Koch}\ \emph {et~al.}(2003)\citenamefont {Koch},
  \citenamefont {de~Mer{\'{a}}s},\ and\ \citenamefont
  {Pedersen}}]{CD_Koch2003}%
  \BibitemOpen
  \bibfield  {author} {\bibinfo {author} {\bibfnamefont {H.}~\bibnamefont
  {Koch}}, \bibinfo {author} {\bibfnamefont {A.~S.}\ \bibnamefont
  {de~Mer{\'{a}}s}},\ and\ \bibinfo {author} {\bibfnamefont {T.~B.}\
  \bibnamefont {Pedersen}},\ }\bibfield  {title} {\bibinfo {title} {Reduced
  scaling in electronic structure calculations using {C}holesky
  decompositions},\ }\href {https://doi.org/10.1063/1.1578621} {\bibfield
  {journal} {\bibinfo  {journal} {J. Chem. Phys.}\ }\textbf {\bibinfo {volume}
  {118}},\ \bibinfo {pages} {9481} (\bibinfo {year} {2003})}\BibitemShut
  {NoStop}%
\bibitem [{\citenamefont {Hohenstein}\ \emph
  {et~al.}(2012{\natexlab{a}})\citenamefont {Hohenstein}, \citenamefont
  {Parrish},\ and\ \citenamefont {Martínez}}]{THC_one_Martinez2012}%
  \BibitemOpen
  \bibfield  {author} {\bibinfo {author} {\bibfnamefont {E.~G.}\ \bibnamefont
  {Hohenstein}}, \bibinfo {author} {\bibfnamefont {R.~M.}\ \bibnamefont
  {Parrish}},\ and\ \bibinfo {author} {\bibfnamefont {T.~J.}\ \bibnamefont
  {Martínez}},\ }\bibfield  {title} {\bibinfo {title} {{Tensor
  hypercontraction density fitting. I. Quartic scaling second- and third-order
  {M{\o}ller-Plesset} perturbation theory}},\ }\href
  {https://doi.org/10.1063/1.4732310} {\bibfield  {journal} {\bibinfo
  {journal} {J. Chem. Phys.}\ }\textbf {\bibinfo {volume} {137}},\ \bibinfo
  {pages} {044103} (\bibinfo {year} {2012}{\natexlab{a}})}\BibitemShut
  {NoStop}%
\bibitem [{\citenamefont {Parrish}\ \emph {et~al.}(2012)\citenamefont
  {Parrish}, \citenamefont {Hohenstein}, \citenamefont {Martínez},\ and\
  \citenamefont {Sherrill}}]{LSTHC_Sherrill2012}%
  \BibitemOpen
  \bibfield  {author} {\bibinfo {author} {\bibfnamefont {R.~M.}\ \bibnamefont
  {Parrish}}, \bibinfo {author} {\bibfnamefont {E.~G.}\ \bibnamefont
  {Hohenstein}}, \bibinfo {author} {\bibfnamefont {T.~J.}\ \bibnamefont
  {Martínez}},\ and\ \bibinfo {author} {\bibfnamefont {C.~D.}\ \bibnamefont
  {Sherrill}},\ }\bibfield  {title} {\bibinfo {title} {Tensor hypercontraction.
  {II}. {L}east-squares renormalization},\ }\href
  {https://doi.org/10.1063/1.4768233} {\bibfield  {journal} {\bibinfo
  {journal} {J. Chem. Phys.}\ }\textbf {\bibinfo {volume} {137}},\ \bibinfo
  {pages} {224106} (\bibinfo {year} {2012})}\BibitemShut {NoStop}%
\bibitem [{\citenamefont {Lu}\ and\ \citenamefont {Ying}(2015)}]{ISDF_Lu2015}%
  \BibitemOpen
  \bibfield  {author} {\bibinfo {author} {\bibfnamefont {J.}~\bibnamefont
  {Lu}}\ and\ \bibinfo {author} {\bibfnamefont {L.}~\bibnamefont {Ying}},\
  }\bibfield  {title} {\bibinfo {title} {Compression of the electron repulsion
  integral tensor in tensor hypercontraction format with cubic scaling cost},\
  }\href {https://doi.org/https://doi.org/10.1016/j.jcp.2015.09.014} {\bibfield
   {journal} {\bibinfo  {journal} {J. Comput. Phys.}\ }\textbf {\bibinfo
  {volume} {302}},\ \bibinfo {pages} {329} (\bibinfo {year}
  {2015})}\BibitemShut {NoStop}%
\bibitem [{\citenamefont {Lu}\ and\ \citenamefont
  {Ying}(2016)}]{ISDF_Bloch_Lu2016}%
  \BibitemOpen
  \bibfield  {author} {\bibinfo {author} {\bibfnamefont {J.}~\bibnamefont
  {Lu}}\ and\ \bibinfo {author} {\bibfnamefont {L.}~\bibnamefont {Ying}},\
  }\bibfield  {title} {\bibinfo {title} {Fast algorithm for periodic density
  fitting for {B}loch waves},\ }\href
  {https://doi.org/10.4310/amsa.2016.v1.n2.a3} {\bibfield  {journal} {\bibinfo
  {journal} {Ann. Math. Sci. Appl.}\ }\textbf {\bibinfo {volume} {1}},\
  \bibinfo {pages} {321} (\bibinfo {year} {2016})}\BibitemShut {NoStop}%
\bibitem [{\citenamefont {Yeh}\ and\ \citenamefont
  {Morales}(2023)}]{THC-RPA_CNY2023}%
  \BibitemOpen
  \bibfield  {author} {\bibinfo {author} {\bibfnamefont {C.-N.}\ \bibnamefont
  {Yeh}}\ and\ \bibinfo {author} {\bibfnamefont {M.~A.}\ \bibnamefont
  {Morales}},\ }\bibfield  {title} {\bibinfo {title} {Low-scaling algorithm for
  the random phase approximation using tensor hypercontraction with k-point
  sampling},\ }\href {https://doi.org/10.1021/acs.jctc.3c00615} {\bibfield
  {journal} {\bibinfo  {journal} {J. Chem. Theory Comput.}\ }\textbf {\bibinfo
  {volume} {19}},\ \bibinfo {pages} {6197} (\bibinfo {year}
  {2023})}\BibitemShut {NoStop}%
\bibitem [{\citenamefont {Chinnamsetty}\ \emph {et~al.}(2007)\citenamefont
  {Chinnamsetty}, \citenamefont {Espig}, \citenamefont {Khoromskij},
  \citenamefont {Hackbusch},\ and\ \citenamefont {Flad}}]{CP_Chinnamsetty2007}%
  \BibitemOpen
  \bibfield  {author} {\bibinfo {author} {\bibfnamefont {S.~R.}\ \bibnamefont
  {Chinnamsetty}}, \bibinfo {author} {\bibfnamefont {M.}~\bibnamefont {Espig}},
  \bibinfo {author} {\bibfnamefont {B.~N.}\ \bibnamefont {Khoromskij}},
  \bibinfo {author} {\bibfnamefont {W.}~\bibnamefont {Hackbusch}},\ and\
  \bibinfo {author} {\bibfnamefont {H.-J.}\ \bibnamefont {Flad}},\ }\bibfield
  {title} {\bibinfo {title} {Tensor product approximation with optimal rank in
  quantum chemistry},\ }\href {https://doi.org/10.1063/1.2761871} {\bibfield
  {journal} {\bibinfo  {journal} {J. Chem. Phys.}\ }\textbf {\bibinfo {volume}
  {127}},\ \bibinfo {pages} {084110} (\bibinfo {year} {2007})}\BibitemShut
  {NoStop}%
\bibitem [{\citenamefont {Benedikt}\ \emph
  {et~al.}(2013{\natexlab{a}})\citenamefont {Benedikt}, \citenamefont {Auer},
  \citenamefont {Espig}, \citenamefont {Hackbusch},\ and\ \citenamefont
  {Auer}}]{CP_Benedikt2013}%
  \BibitemOpen
  \bibfield  {author} {\bibinfo {author} {\bibfnamefont {U.}~\bibnamefont
  {Benedikt}}, \bibinfo {author} {\bibfnamefont {H.}~\bibnamefont {Auer}},
  \bibinfo {author} {\bibfnamefont {M.}~\bibnamefont {Espig}}, \bibinfo
  {author} {\bibfnamefont {W.}~\bibnamefont {Hackbusch}},\ and\ \bibinfo
  {author} {\bibfnamefont {A.~A.}\ \bibnamefont {Auer}},\ }\bibfield  {title}
  {\bibinfo {title} {Tensor representation techniques in {post-Hartree–Fock}
  methods: matrix product state tensor format},\ }\href
  {https://doi.org/10.1080/00268976.2013.798433} {\bibfield  {journal}
  {\bibinfo  {journal} {Mol. Phys.}\ }\textbf {\bibinfo {volume} {111}},\
  \bibinfo {pages} {2398} (\bibinfo {year} {2013}{\natexlab{a}})}\BibitemShut
  {NoStop}%
\bibitem [{\citenamefont {Benedikt}\ \emph
  {et~al.}(2013{\natexlab{b}})\citenamefont {Benedikt}, \citenamefont {Böhm},\
  and\ \citenamefont {Auer}}]{CP2_Benedikt2013}%
  \BibitemOpen
  \bibfield  {author} {\bibinfo {author} {\bibfnamefont {U.}~\bibnamefont
  {Benedikt}}, \bibinfo {author} {\bibfnamefont {K.-H.}\ \bibnamefont
  {Böhm}},\ and\ \bibinfo {author} {\bibfnamefont {A.~A.}\ \bibnamefont
  {Auer}},\ }\bibfield  {title} {\bibinfo {title} {Tensor decomposition in
  {post-Hartree–Fock} methods. {II}. {CCD} implementation},\ }\href
  {https://doi.org/10.1063/1.4833565} {\bibfield  {journal} {\bibinfo
  {journal} {J. Chem. Phys.}\ }\textbf {\bibinfo {volume} {139}},\ \bibinfo
  {pages} {224101} (\bibinfo {year} {2013}{\natexlab{b}})}\BibitemShut
  {NoStop}%
\bibitem [{\citenamefont {Böhm}\ \emph {et~al.}(2016)\citenamefont {Böhm},
  \citenamefont {Auer},\ and\ \citenamefont {Espig}}]{CP_Bohm2016}%
  \BibitemOpen
  \bibfield  {author} {\bibinfo {author} {\bibfnamefont {K.-H.}\ \bibnamefont
  {Böhm}}, \bibinfo {author} {\bibfnamefont {A.~A.}\ \bibnamefont {Auer}},\
  and\ \bibinfo {author} {\bibfnamefont {M.}~\bibnamefont {Espig}},\ }\bibfield
   {title} {\bibinfo {title} {Tensor representation techniques for full
  configuration interaction: A {F}ock space approach using the canonical
  product format},\ }\href {https://doi.org/10.1063/1.4953665} {\bibfield
  {journal} {\bibinfo  {journal} {J. Chem. Phys.}\ }\textbf {\bibinfo {volume}
  {144}},\ \bibinfo {pages} {244102} (\bibinfo {year} {2016})}\BibitemShut
  {NoStop}%
\bibitem [{\citenamefont {Schutski}\ \emph {et~al.}(2017)\citenamefont
  {Schutski}, \citenamefont {Zhao}, \citenamefont {Henderson},\ and\
  \citenamefont {Scuseria}}]{CPCC_Schutski2017}%
  \BibitemOpen
  \bibfield  {author} {\bibinfo {author} {\bibfnamefont {R.}~\bibnamefont
  {Schutski}}, \bibinfo {author} {\bibfnamefont {J.}~\bibnamefont {Zhao}},
  \bibinfo {author} {\bibfnamefont {T.~M.}\ \bibnamefont {Henderson}},\ and\
  \bibinfo {author} {\bibfnamefont {G.~E.}\ \bibnamefont {Scuseria}},\
  }\bibfield  {title} {\bibinfo {title} {Tensor-structured coupled cluster
  theory},\ }\href {https://doi.org/10.1063/1.4996988} {\bibfield  {journal}
  {\bibinfo  {journal} {J. Chem. Phys.}\ }\textbf {\bibinfo {volume} {147}},\
  \bibinfo {pages} {184113} (\bibinfo {year} {2017})}\BibitemShut {NoStop}%
\bibitem [{\citenamefont {Xing}\ and\ \citenamefont {Chow}(2020)}]{chow2020}%
  \BibitemOpen
  \bibfield  {author} {\bibinfo {author} {\bibfnamefont {X.}~\bibnamefont
  {Xing}}\ and\ \bibinfo {author} {\bibfnamefont {E.}~\bibnamefont {Chow}},\
  }\bibfield  {title} {\bibinfo {title} {Fast {C}oulomb matrix construction via
  compressing the interactions between continuous charge distributions},\
  }\href {https://doi.org/10.1137/19M1252855} {\bibfield  {journal} {\bibinfo
  {journal} {SIAM J. Sci. Comput.}\ }\textbf {\bibinfo {volume} {42}},\
  \bibinfo {pages} {A162} (\bibinfo {year} {2020})}\BibitemShut {NoStop}%
\bibitem [{\citenamefont {Xing}\ \emph {et~al.}(2020)\citenamefont {Xing},
  \citenamefont {Huang},\ and\ \citenamefont {Chow}}]{Xing_JCP2020}%
  \BibitemOpen
  \bibfield  {author} {\bibinfo {author} {\bibfnamefont {X.}~\bibnamefont
  {Xing}}, \bibinfo {author} {\bibfnamefont {H.}~\bibnamefont {Huang}},\ and\
  \bibinfo {author} {\bibfnamefont {E.}~\bibnamefont {Chow}},\ }\bibfield
  {title} {\bibinfo {title} {A linear scaling hierarchical block low-rank
  representation of the electron repulsion integral tensor},\ }\href
  {https://doi.org/10.1063/5.0010732} {\bibfield  {journal} {\bibinfo
  {journal} {J. Chem. Phys.}\ }\textbf {\bibinfo {volume} {153}},\ \bibinfo
  {pages} {084119} (\bibinfo {year} {2020})}\BibitemShut {NoStop}%
\bibitem [{\citenamefont {Pan}\ and\ \citenamefont {Lindsey}(2025)}]{pan25}%
  \BibitemOpen
  \bibfield  {author} {\bibinfo {author} {\bibfnamefont {Y.}~\bibnamefont
  {Pan}}\ and\ \bibinfo {author} {\bibfnamefont {M.}~\bibnamefont {Lindsey}},\
  }\href {https://arxiv.org/abs/2510.21213} {\bibinfo {title} {Fast adaptive
  discontinuous basis sets for electronic structure}} (\bibinfo {year}
  {2025}),\ \Eprint {https://arxiv.org/abs/2510.21213} {arXiv:2510.21213
  [physics.comp-ph]} \BibitemShut {NoStop}%
\bibitem [{\citenamefont {Jolly}\ \emph {et~al.}(2025)\citenamefont {Jolly},
  \citenamefont {Fern\'andez},\ and\ \citenamefont
  {Waintal}}]{Nicolas_PRB2025}%
  \BibitemOpen
  \bibfield  {author} {\bibinfo {author} {\bibfnamefont {N.}~\bibnamefont
  {Jolly}}, \bibinfo {author} {\bibfnamefont {Y.~N.}\ \bibnamefont
  {Fern\'andez}},\ and\ \bibinfo {author} {\bibfnamefont {X.}~\bibnamefont
  {Waintal}},\ }\bibfield  {title} {\bibinfo {title} {Tensorized orbitals for
  computational chemistry},\ }\href
  {https://doi.org/10.1103/PhysRevB.111.245115} {\bibfield  {journal} {\bibinfo
   {journal} {Phys. Rev. B}\ }\textbf {\bibinfo {volume} {111}},\ \bibinfo
  {pages} {245115} (\bibinfo {year} {2025})}\BibitemShut {NoStop}%
\bibitem [{\citenamefont {Hu}\ \emph {et~al.}(2017)\citenamefont {Hu},
  \citenamefont {Lin},\ and\ \citenamefont {Yang}}]{ISDF_QRCP_hybrid_Hu2017}%
  \BibitemOpen
  \bibfield  {author} {\bibinfo {author} {\bibfnamefont {W.}~\bibnamefont
  {Hu}}, \bibinfo {author} {\bibfnamefont {L.}~\bibnamefont {Lin}},\ and\
  \bibinfo {author} {\bibfnamefont {C.}~\bibnamefont {Yang}},\ }\bibfield
  {title} {\bibinfo {title} {Interpolative separable density fitting
  decomposition for accelerating hybrid density functional calculations with
  applications to defects in silicon},\ }\href
  {https://doi.org/10.1021/acs.jctc.7b00807} {\bibfield  {journal} {\bibinfo
  {journal} {J. Chem. Theory Comput.}\ }\textbf {\bibinfo {volume} {13}},\
  \bibinfo {pages} {5420} (\bibinfo {year} {2017})}\BibitemShut {NoStop}%
\bibitem [{\citenamefont {Dong}\ \emph {et~al.}(2018)\citenamefont {Dong},
  \citenamefont {Hu},\ and\ \citenamefont {Lin}}]{ISDF_CVT_hybrid_Lin2018}%
  \BibitemOpen
  \bibfield  {author} {\bibinfo {author} {\bibfnamefont {K.}~\bibnamefont
  {Dong}}, \bibinfo {author} {\bibfnamefont {W.}~\bibnamefont {Hu}},\ and\
  \bibinfo {author} {\bibfnamefont {L.}~\bibnamefont {Lin}},\ }\bibfield
  {title} {\bibinfo {title} {Interpolative separable density fitting through
  centroidal {V}oronoi tessellation with applications to hybrid functional
  electronic structure calculations},\ }\href
  {https://doi.org/10.1021/acs.jctc.7b01113} {\bibfield  {journal} {\bibinfo
  {journal} {J. Chem. Theory Comput.}\ }\textbf {\bibinfo {volume} {14}},\
  \bibinfo {pages} {1311} (\bibinfo {year} {2018})}\BibitemShut {NoStop}%
\bibitem [{\citenamefont {Qin}\ \emph {et~al.}(2020{\natexlab{a}})\citenamefont
  {Qin}, \citenamefont {Liu}, \citenamefont {Hu},\ and\ \citenamefont
  {Yang}}]{ISDF_hybridDFT_NAO_Qin2020}%
  \BibitemOpen
  \bibfield  {author} {\bibinfo {author} {\bibfnamefont {X.}~\bibnamefont
  {Qin}}, \bibinfo {author} {\bibfnamefont {J.}~\bibnamefont {Liu}}, \bibinfo
  {author} {\bibfnamefont {W.}~\bibnamefont {Hu}},\ and\ \bibinfo {author}
  {\bibfnamefont {J.}~\bibnamefont {Yang}},\ }\bibfield  {title} {\bibinfo
  {title} {Interpolative separable density fitting decomposition for
  accelerating {H}artree–{F}ock exchange calculations within numerical atomic
  orbitals},\ }\href {https://doi.org/10.1021/acs.jpca.0c02826} {\bibfield
  {journal} {\bibinfo  {journal} {J. Phys. Chem. A}\ }\textbf {\bibinfo
  {volume} {124}},\ \bibinfo {pages} {5664} (\bibinfo {year}
  {2020}{\natexlab{a}})}\BibitemShut {NoStop}%
\bibitem [{\citenamefont {Qin}\ \emph {et~al.}(2020{\natexlab{b}})\citenamefont
  {Qin}, \citenamefont {Li}, \citenamefont {Hu},\ and\ \citenamefont
  {Yang}}]{ISDF_CVT_hybrid_NAO_Qin2020}%
  \BibitemOpen
  \bibfield  {author} {\bibinfo {author} {\bibfnamefont {X.}~\bibnamefont
  {Qin}}, \bibinfo {author} {\bibfnamefont {J.}~\bibnamefont {Li}}, \bibinfo
  {author} {\bibfnamefont {W.}~\bibnamefont {Hu}},\ and\ \bibinfo {author}
  {\bibfnamefont {J.}~\bibnamefont {Yang}},\ }\bibfield  {title} {\bibinfo
  {title} {Machine learning k-means clustering algorithm for interpolative
  separable density fitting to accelerate hybrid functional calculations with
  numerical atomic orbitals},\ }\href
  {https://doi.org/10.1021/acs.jpca.0c06019} {\bibfield  {journal} {\bibinfo
  {journal} {J. Phys. Chem. A}\ }\textbf {\bibinfo {volume} {124}},\ \bibinfo
  {pages} {10066} (\bibinfo {year} {2020}{\natexlab{b}})}\BibitemShut {NoStop}%
\bibitem [{\citenamefont {Sharma}\ \emph {et~al.}(2022)\citenamefont {Sharma},
  \citenamefont {White},\ and\ \citenamefont {Beylkin}}]{rPS_HF_Sharma2022}%
  \BibitemOpen
  \bibfield  {author} {\bibinfo {author} {\bibfnamefont {S.}~\bibnamefont
  {Sharma}}, \bibinfo {author} {\bibfnamefont {A.~F.}\ \bibnamefont {White}},\
  and\ \bibinfo {author} {\bibfnamefont {G.}~\bibnamefont {Beylkin}},\
  }\bibfield  {title} {\bibinfo {title} {Fast exchange with {G}aussian basis
  set using robust pseudospectral method},\ }\href
  {https://doi.org/10.1021/acs.jctc.2c00720} {\bibfield  {journal} {\bibinfo
  {journal} {J. Chem. Theory Comput.}\ }\textbf {\bibinfo {volume} {18}},\
  \bibinfo {pages} {7306} (\bibinfo {year} {2022})}\BibitemShut {NoStop}%
\bibitem [{\citenamefont {Hohenstein}\ \emph
  {et~al.}(2012{\natexlab{b}})\citenamefont {Hohenstein}, \citenamefont
  {Parrish}, \citenamefont {Sherrill},\ and\ \citenamefont
  {Martínez}}]{THC_three_Martinez2012}%
  \BibitemOpen
  \bibfield  {author} {\bibinfo {author} {\bibfnamefont {E.~G.}\ \bibnamefont
  {Hohenstein}}, \bibinfo {author} {\bibfnamefont {R.~M.}\ \bibnamefont
  {Parrish}}, \bibinfo {author} {\bibfnamefont {C.~D.}\ \bibnamefont
  {Sherrill}},\ and\ \bibinfo {author} {\bibfnamefont {T.~J.}\ \bibnamefont
  {Martínez}},\ }\bibfield  {title} {\bibinfo {title} {{Communication: Tensor
  hypercontraction. III. Least-squares tensor hypercontraction for the
  determination of correlated wavefunctions}},\ }\href
  {https://doi.org/10.1063/1.4768241} {\bibfield  {journal} {\bibinfo
  {journal} {J. Chem. Phys.}\ }\textbf {\bibinfo {volume} {137}},\ \bibinfo
  {pages} {221101} (\bibinfo {year} {2012}{\natexlab{b}})}\BibitemShut
  {NoStop}%
\bibitem [{\citenamefont {Hohenstein}\ \emph
  {et~al.}(2013{\natexlab{a}})\citenamefont {Hohenstein}, \citenamefont
  {Kokkila}, \citenamefont {Parrish},\ and\ \citenamefont
  {Martínez}}]{THC_CC2_Hohenstein2013}%
  \BibitemOpen
  \bibfield  {author} {\bibinfo {author} {\bibfnamefont {E.~G.}\ \bibnamefont
  {Hohenstein}}, \bibinfo {author} {\bibfnamefont {S.~I.~L.}\ \bibnamefont
  {Kokkila}}, \bibinfo {author} {\bibfnamefont {R.~M.}\ \bibnamefont
  {Parrish}},\ and\ \bibinfo {author} {\bibfnamefont {T.~J.}\ \bibnamefont
  {Martínez}},\ }\bibfield  {title} {\bibinfo {title} {{Quartic scaling
  second-order approximate coupled cluster singles and doubles via tensor
  hypercontraction: THC-CC2}},\ }\href {https://doi.org/10.1063/1.4795514}
  {\bibfield  {journal} {\bibinfo  {journal} {J. Chem. Phys.}\ }\textbf
  {\bibinfo {volume} {138}},\ \bibinfo {pages} {124111} (\bibinfo {year}
  {2013}{\natexlab{a}})}\BibitemShut {NoStop}%
\bibitem [{\citenamefont {Hohenstein}\ \emph
  {et~al.}(2013{\natexlab{b}})\citenamefont {Hohenstein}, \citenamefont
  {Kokkila}, \citenamefont {Parrish},\ and\ \citenamefont
  {Martínez}}]{THC_EOMCC2_Hohenstein2013}%
  \BibitemOpen
  \bibfield  {author} {\bibinfo {author} {\bibfnamefont {E.~G.}\ \bibnamefont
  {Hohenstein}}, \bibinfo {author} {\bibfnamefont {S.~I.~L.}\ \bibnamefont
  {Kokkila}}, \bibinfo {author} {\bibfnamefont {R.~M.}\ \bibnamefont
  {Parrish}},\ and\ \bibinfo {author} {\bibfnamefont {T.~J.}\ \bibnamefont
  {Martínez}},\ }\bibfield  {title} {\bibinfo {title} {Tensor hypercontraction
  equation-of-motion second-order approximate coupled cluster: Electronic
  excitation energies in {$O(N^4)$} time},\ }\href
  {https://doi.org/10.1021/jp4021905} {\bibfield  {journal} {\bibinfo
  {journal} {J. Phys. Chem. B}\ }\textbf {\bibinfo {volume} {117}},\ \bibinfo
  {pages} {12972} (\bibinfo {year} {2013}{\natexlab{b}})}\BibitemShut {NoStop}%
\bibitem [{\citenamefont {Parrish}\ \emph {et~al.}(2014)\citenamefont
  {Parrish}, \citenamefont {Sherrill}, \citenamefont {Hohenstein},
  \citenamefont {Kokkila},\ and\ \citenamefont
  {Martínez}}]{LSTHC_CCSD_Parrish2014}%
  \BibitemOpen
  \bibfield  {author} {\bibinfo {author} {\bibfnamefont {R.~M.}\ \bibnamefont
  {Parrish}}, \bibinfo {author} {\bibfnamefont {C.~D.}\ \bibnamefont
  {Sherrill}}, \bibinfo {author} {\bibfnamefont {E.~G.}\ \bibnamefont
  {Hohenstein}}, \bibinfo {author} {\bibfnamefont {S.~I.~L.}\ \bibnamefont
  {Kokkila}},\ and\ \bibinfo {author} {\bibfnamefont {T.~J.}\ \bibnamefont
  {Martínez}},\ }\bibfield  {title} {\bibinfo {title} {{Communication:
  Acceleration of coupled cluster singles and doubles via orbital-weighted
  least-squares tensor hypercontraction}},\ }\href
  {https://doi.org/10.1063/1.4876016} {\bibfield  {journal} {\bibinfo
  {journal} {J. Chem. Phys.}\ }\textbf {\bibinfo {volume} {140}},\ \bibinfo
  {pages} {181102} (\bibinfo {year} {2014})}\BibitemShut {NoStop}%
\bibitem [{\citenamefont {Hohenstein}\ \emph {et~al.}(2022)\citenamefont
  {Hohenstein}, \citenamefont {Fales}, \citenamefont {Parrish},\ and\
  \citenamefont {Mart{\'{\i}}nez}}]{THC_CC_Hohenstein2022}%
  \BibitemOpen
  \bibfield  {author} {\bibinfo {author} {\bibfnamefont {E.~G.}\ \bibnamefont
  {Hohenstein}}, \bibinfo {author} {\bibfnamefont {B.~S.}\ \bibnamefont
  {Fales}}, \bibinfo {author} {\bibfnamefont {R.~M.}\ \bibnamefont {Parrish}},\
  and\ \bibinfo {author} {\bibfnamefont {T.~J.}\ \bibnamefont
  {Mart{\'{\i}}nez}},\ }\bibfield  {title} {\bibinfo {title} {Rank-reduced
  coupled-cluster. {III}. {T}ensor hypercontraction of the doubles
  amplitudes},\ }\href {https://doi.org/10.1063/5.0077770} {\bibfield
  {journal} {\bibinfo  {journal} {J. Chem. Phys.}\ }\textbf {\bibinfo {volume}
  {156}},\ \bibinfo {pages} {054102} (\bibinfo {year} {2022})}\BibitemShut
  {NoStop}%
\bibitem [{\citenamefont {Parrish}\ \emph {et~al.}(2013)\citenamefont
  {Parrish}, \citenamefont {Hohenstein}, \citenamefont {Martínez},\ and\
  \citenamefont {Sherrill}}]{THC_DVR_Sherrill2013}%
  \BibitemOpen
  \bibfield  {author} {\bibinfo {author} {\bibfnamefont {R.~M.}\ \bibnamefont
  {Parrish}}, \bibinfo {author} {\bibfnamefont {E.~G.}\ \bibnamefont
  {Hohenstein}}, \bibinfo {author} {\bibfnamefont {T.~J.}\ \bibnamefont
  {Martínez}},\ and\ \bibinfo {author} {\bibfnamefont {C.~D.}\ \bibnamefont
  {Sherrill}},\ }\bibfield  {title} {\bibinfo {title} {{Discrete variable
  representation in electronic structure theory: Quadrature grids for
  least-squares tensor hypercontraction}},\ }\href
  {https://doi.org/10.1063/1.4802773} {\bibfield  {journal} {\bibinfo
  {journal} {J. Chem. Phys.}\ }\textbf {\bibinfo {volume} {138}},\ \bibinfo
  {pages} {194107} (\bibinfo {year} {2013})}\BibitemShut {NoStop}%
\bibitem [{\citenamefont {Kokkila~Schumacher}\ \emph
  {et~al.}(2015)\citenamefont {Kokkila~Schumacher}, \citenamefont {Hohenstein},
  \citenamefont {Parrish}, \citenamefont {Wang},\ and\ \citenamefont
  {Martínez}}]{LSTHC_MP2_Schumacher2015}%
  \BibitemOpen
  \bibfield  {author} {\bibinfo {author} {\bibfnamefont {S.~I.~L.}\
  \bibnamefont {Kokkila~Schumacher}}, \bibinfo {author} {\bibfnamefont {E.~G.}\
  \bibnamefont {Hohenstein}}, \bibinfo {author} {\bibfnamefont {R.~M.}\
  \bibnamefont {Parrish}}, \bibinfo {author} {\bibfnamefont {L.-P.}\
  \bibnamefont {Wang}},\ and\ \bibinfo {author} {\bibfnamefont {T.~J.}\
  \bibnamefont {Martínez}},\ }\bibfield  {title} {\bibinfo {title} {Tensor
  hypercontraction second-order {M{\o}ller-Plesset} perturbation theory: Grid
  optimization and reaction energies},\ }\href
  {https://doi.org/10.1021/acs.jctc.5b00272} {\bibfield  {journal} {\bibinfo
  {journal} {J. Chem. Theory Comput.}\ }\textbf {\bibinfo {volume} {11}},\
  \bibinfo {pages} {3042} (\bibinfo {year} {2015})}\BibitemShut {NoStop}%
\bibitem [{\citenamefont {Song}\ and\ \citenamefont
  {Martínez}(2016)}]{THC_SOS_MP2_one_Song2016}%
  \BibitemOpen
  \bibfield  {author} {\bibinfo {author} {\bibfnamefont {C.}~\bibnamefont
  {Song}}\ and\ \bibinfo {author} {\bibfnamefont {T.~J.}\ \bibnamefont
  {Martínez}},\ }\bibfield  {title} {\bibinfo {title} {{Atomic orbital-based
  SOS-MP2 with tensor hypercontraction. I. GPU-based tensor construction and
  exploiting sparsity}},\ }\href {https://doi.org/10.1063/1.4948438} {\bibfield
   {journal} {\bibinfo  {journal} {J. Chem. Phys.}\ }\textbf {\bibinfo {volume}
  {144}},\ \bibinfo {pages} {174111} (\bibinfo {year} {2016})}\BibitemShut
  {NoStop}%
\bibitem [{\citenamefont {Song}\ and\ \citenamefont
  {Martínez}(2017{\natexlab{a}})}]{THC_SOS_MP2_two_Song2017}%
  \BibitemOpen
  \bibfield  {author} {\bibinfo {author} {\bibfnamefont {C.}~\bibnamefont
  {Song}}\ and\ \bibinfo {author} {\bibfnamefont {T.~J.}\ \bibnamefont
  {Martínez}},\ }\bibfield  {title} {\bibinfo {title} {{Atomic orbital-based
  SOS-MP2 with tensor hypercontraction. II. Local tensor hypercontraction}},\
  }\href {https://doi.org/10.1063/1.4973840} {\bibfield  {journal} {\bibinfo
  {journal} {J. Chem. Phys.}\ }\textbf {\bibinfo {volume} {146}},\ \bibinfo
  {pages} {034104} (\bibinfo {year} {2017}{\natexlab{a}})}\BibitemShut
  {NoStop}%
\bibitem [{\citenamefont {Song}\ and\ \citenamefont
  {Martínez}(2017{\natexlab{b}})}]{THC_SOS_MP2_gradient_Song2017}%
  \BibitemOpen
  \bibfield  {author} {\bibinfo {author} {\bibfnamefont {C.}~\bibnamefont
  {Song}}\ and\ \bibinfo {author} {\bibfnamefont {T.~J.}\ \bibnamefont
  {Martínez}},\ }\bibfield  {title} {\bibinfo {title} {{Analytical gradients
  for tensor hyper-contracted MP2 and SOS-MP2 on graphical processing units}},\
  }\href {https://doi.org/10.1063/1.4997997} {\bibfield  {journal} {\bibinfo
  {journal} {J. Chem. Phys.}\ }\textbf {\bibinfo {volume} {147}},\ \bibinfo
  {pages} {161723} (\bibinfo {year} {2017}{\natexlab{b}})}\BibitemShut
  {NoStop}%
\bibitem [{\citenamefont {Lee}\ \emph {et~al.}(2020)\citenamefont {Lee},
  \citenamefont {Lin},\ and\ \citenamefont {Head-Gordon}}]{THC_MP3_Joonho2020}%
  \BibitemOpen
  \bibfield  {author} {\bibinfo {author} {\bibfnamefont {J.}~\bibnamefont
  {Lee}}, \bibinfo {author} {\bibfnamefont {L.}~\bibnamefont {Lin}},\ and\
  \bibinfo {author} {\bibfnamefont {M.}~\bibnamefont {Head-Gordon}},\
  }\bibfield  {title} {\bibinfo {title} {Systematically improvable tensor
  hypercontraction: Interpolative separable density-fitting for molecules
  applied to exact exchange, second- and third-order {M{\o}ller-Plesset}
  perturbation theory},\ }\href {https://doi.org/10.1021/acs.jctc.9b00820}
  {\bibfield  {journal} {\bibinfo  {journal} {J. Chem. Theory Comput.}\
  }\textbf {\bibinfo {volume} {16}},\ \bibinfo {pages} {243} (\bibinfo {year}
  {2020})}\BibitemShut {NoStop}%
\bibitem [{\citenamefont {Matthews}(2021)}]{THC_MP3_Matthews2021}%
  \BibitemOpen
  \bibfield  {author} {\bibinfo {author} {\bibfnamefont {D.~A.}\ \bibnamefont
  {Matthews}},\ }\bibfield  {title} {\bibinfo {title} {A critical analysis of
  least-squares tensor hypercontraction applied to {MP}3},\ }\href
  {https://doi.org/10.1063/5.0038764} {\bibfield  {journal} {\bibinfo
  {journal} {J. Chem. Phys.}\ }\textbf {\bibinfo {volume} {154}},\ \bibinfo
  {pages} {134102} (\bibinfo {year} {2021})}\BibitemShut {NoStop}%
\bibitem [{\citenamefont {Gao}\ and\ \citenamefont
  {Chelikowsky}(2020)}]{TDDFT_G0W0_ISDF_Gao2020}%
  \BibitemOpen
  \bibfield  {author} {\bibinfo {author} {\bibfnamefont {W.}~\bibnamefont
  {Gao}}\ and\ \bibinfo {author} {\bibfnamefont {J.~R.}\ \bibnamefont
  {Chelikowsky}},\ }\bibfield  {title} {\bibinfo {title} {Accelerating
  time-dependent density functional theory and {$GW$} calculations for
  molecules and nanoclusters with symmetry adapted interpolative separable
  density fitting},\ }\href {https://doi.org/10.1021/acs.jctc.9b01025}
  {\bibfield  {journal} {\bibinfo  {journal} {J. Chem. Theory Comput.}\
  }\textbf {\bibinfo {volume} {16}},\ \bibinfo {pages} {2216} (\bibinfo {year}
  {2020})}\BibitemShut {NoStop}%
\bibitem [{\citenamefont {Ma}\ \emph {et~al.}(2021)\citenamefont {Ma},
  \citenamefont {Wang}, \citenamefont {Wan}, \citenamefont {Li}, \citenamefont
  {Qin}, \citenamefont {Liu}, \citenamefont {Hu}, \citenamefont {Lin},
  \citenamefont {Yang},\ and\ \citenamefont {Yang}}]{G0W0_COHSEX_ISDF_Ma2021}%
  \BibitemOpen
  \bibfield  {author} {\bibinfo {author} {\bibfnamefont {H.}~\bibnamefont
  {Ma}}, \bibinfo {author} {\bibfnamefont {L.}~\bibnamefont {Wang}}, \bibinfo
  {author} {\bibfnamefont {L.}~\bibnamefont {Wan}}, \bibinfo {author}
  {\bibfnamefont {J.}~\bibnamefont {Li}}, \bibinfo {author} {\bibfnamefont
  {X.}~\bibnamefont {Qin}}, \bibinfo {author} {\bibfnamefont {J.}~\bibnamefont
  {Liu}}, \bibinfo {author} {\bibfnamefont {W.}~\bibnamefont {Hu}}, \bibinfo
  {author} {\bibfnamefont {L.}~\bibnamefont {Lin}}, \bibinfo {author}
  {\bibfnamefont {C.}~\bibnamefont {Yang}},\ and\ \bibinfo {author}
  {\bibfnamefont {J.}~\bibnamefont {Yang}},\ }\bibfield  {title} {\bibinfo
  {title} {Realizing effective cubic-scaling {C}oulomb hole plus screened
  exchange approximation in periodic systems via interpolative separable
  density fitting with a plane-wave basis set},\ }\href
  {https://doi.org/10.1021/acs.jpca.1c03762} {\bibfield  {journal} {\bibinfo
  {journal} {J. Phys. Chem. A}\ }\textbf {\bibinfo {volume} {125}},\ \bibinfo
  {pages} {7545} (\bibinfo {year} {2021})}\BibitemShut {NoStop}%
\bibitem [{\citenamefont {Duchemin}\ and\ \citenamefont
  {Blase}(2021)}]{separable_RI_G0W0_Blase2021}%
  \BibitemOpen
  \bibfield  {author} {\bibinfo {author} {\bibfnamefont {I.}~\bibnamefont
  {Duchemin}}\ and\ \bibinfo {author} {\bibfnamefont {X.}~\bibnamefont
  {Blase}},\ }\bibfield  {title} {\bibinfo {title} {Cubic-scaling all-electron
  {$GW$} calculations with a separable density-fitting space–time approach},\
  }\href {https://doi.org/10.1021/acs.jctc.1c00101} {\bibfield  {journal}
  {\bibinfo  {journal} {J. Chem. Theory Comput.}\ }\textbf {\bibinfo {volume}
  {17}},\ \bibinfo {pages} {2383} (\bibinfo {year} {2021})}\BibitemShut
  {NoStop}%
\bibitem [{\citenamefont {Malone}\ \emph {et~al.}(2019)\citenamefont {Malone},
  \citenamefont {Zhang},\ and\ \citenamefont
  {Morales}}]{AFQMC_ISDF_Miguel2019}%
  \BibitemOpen
  \bibfield  {author} {\bibinfo {author} {\bibfnamefont {F.~D.}\ \bibnamefont
  {Malone}}, \bibinfo {author} {\bibfnamefont {S.}~\bibnamefont {Zhang}},\ and\
  \bibinfo {author} {\bibfnamefont {M.~A.}\ \bibnamefont {Morales}},\
  }\bibfield  {title} {\bibinfo {title} {Overcoming the memory bottleneck in
  auxiliary field quantum {M}onte {C}arlo simulations with interpolative
  separable density fitting},\ }\href
  {https://doi.org/10.1021/acs.jctc.8b00944} {\bibfield  {journal} {\bibinfo
  {journal} {J. Chem. Theory Comput.}\ }\textbf {\bibinfo {volume} {15}},\
  \bibinfo {pages} {256} (\bibinfo {year} {2019})}\BibitemShut {NoStop}%
\bibitem [{\citenamefont {Cheng}\ \emph {et~al.}(2005)\citenamefont {Cheng},
  \citenamefont {Gimbutas}, \citenamefont {Martinsson},\ and\ \citenamefont
  {Rokhlin}}]{ID_Cheng2005}%
  \BibitemOpen
  \bibfield  {author} {\bibinfo {author} {\bibfnamefont {H.}~\bibnamefont
  {Cheng}}, \bibinfo {author} {\bibfnamefont {Z.}~\bibnamefont {Gimbutas}},
  \bibinfo {author} {\bibfnamefont {P.~G.}\ \bibnamefont {Martinsson}},\ and\
  \bibinfo {author} {\bibfnamefont {V.}~\bibnamefont {Rokhlin}},\ }\bibfield
  {title} {\bibinfo {title} {On the compression of low rank matrices},\ }\href
  {https://doi.org/10.1137/030602678} {\bibfield  {journal} {\bibinfo
  {journal} {SIAM J. Sci. Comput.}\ }\textbf {\bibinfo {volume} {26}},\
  \bibinfo {pages} {1389} (\bibinfo {year} {2005})}\BibitemShut {NoStop}%
\bibitem [{\citenamefont {Liberty}\ \emph {et~al.}(2007)\citenamefont
  {Liberty}, \citenamefont {Woolfe}, \citenamefont {Martinsson}, \citenamefont
  {Rokhlin},\ and\ \citenamefont {Tygert}}]{ID_Liberty2007}%
  \BibitemOpen
  \bibfield  {author} {\bibinfo {author} {\bibfnamefont {E.}~\bibnamefont
  {Liberty}}, \bibinfo {author} {\bibfnamefont {F.}~\bibnamefont {Woolfe}},
  \bibinfo {author} {\bibfnamefont {P.-G.}\ \bibnamefont {Martinsson}},
  \bibinfo {author} {\bibfnamefont {V.}~\bibnamefont {Rokhlin}},\ and\ \bibinfo
  {author} {\bibfnamefont {M.}~\bibnamefont {Tygert}},\ }\bibfield  {title}
  {\bibinfo {title} {Randomized algorithms for the low-rank approximation of
  matrices},\ }\href {https://doi.org/10.1073/pnas.0709640104} {\bibfield
  {journal} {\bibinfo  {journal} {Proc. Natl. Acad. Sci.}\ }\textbf {\bibinfo
  {volume} {104}},\ \bibinfo {pages} {20167} (\bibinfo {year}
  {2007})}\BibitemShut {NoStop}%
\bibitem [{\citenamefont {Jiang}\ and\ \citenamefont
  {Greengard}(2025)}]{jiang2025cpam}%
  \BibitemOpen
  \bibfield  {author} {\bibinfo {author} {\bibfnamefont {S.}~\bibnamefont
  {Jiang}}\ and\ \bibinfo {author} {\bibfnamefont {L.}~\bibnamefont
  {Greengard}},\ }\bibfield  {title} {\bibinfo {title} {A dual-space multilevel
  kernel-splitting framework for discrete and continuous convolution},\ }\href
  {https://doi.org/10.1002/cpa.22240} {\bibfield  {journal} {\bibinfo
  {journal} {Comm. Pure Appl. Math.}\ }\textbf {\bibinfo {volume} {78}},\
  \bibinfo {pages} {1086} (\bibinfo {year} {2025})}\BibitemShut {NoStop}%
\bibitem [{\citenamefont {Yeh}\ and\ \citenamefont
  {Morales}(2024)}]{Yeh_ISDF_GW2024}%
  \BibitemOpen
  \bibfield  {author} {\bibinfo {author} {\bibfnamefont {C.-N.}\ \bibnamefont
  {Yeh}}\ and\ \bibinfo {author} {\bibfnamefont {M.~A.}\ \bibnamefont
  {Morales}},\ }\bibfield  {title} {\bibinfo {title} {Low-scaling algorithms
  for {$GW$} and constrained random phase approximation using symmetry-adapted
  interpolative separable density fitting},\ }\href
  {https://doi.org/10.1021/acs.jctc.4c00085} {\bibfield  {journal} {\bibinfo
  {journal} {J. Chem. Theory Comput.}\ }\textbf {\bibinfo {volume} {20}},\
  \bibinfo {pages} {3184} (\bibinfo {year} {2024})}\BibitemShut {NoStop}%
\bibitem [{\citenamefont {Matthews}(2020)}]{LSTHC_Matthews2020}%
  \BibitemOpen
  \bibfield  {author} {\bibinfo {author} {\bibfnamefont {D.~A.}\ \bibnamefont
  {Matthews}},\ }\bibfield  {title} {\bibinfo {title} {Improved grid
  optimization and fitting in least squares tensor hypercontraction},\ }\href
  {https://doi.org/10.1021/acs.jctc.9b01205} {\bibfield  {journal} {\bibinfo
  {journal} {J. Chem. Theory Comput.}\ }\textbf {\bibinfo {volume} {16}},\
  \bibinfo {pages} {1382} (\bibinfo {year} {2020})}\BibitemShut {NoStop}%
\bibitem [{\citenamefont {Trefethen}\ and\ \citenamefont
  {Weideman}(2014)}]{trefethen2014sirev}%
  \BibitemOpen
  \bibfield  {author} {\bibinfo {author} {\bibfnamefont {L.~N.}\ \bibnamefont
  {Trefethen}}\ and\ \bibinfo {author} {\bibfnamefont {J.}~\bibnamefont
  {Weideman}},\ }\bibfield  {title} {\bibinfo {title} {The exponentially
  convergent trapezoidal rule},\ }\href {https://doi.org/10.1137/130932132}
  {\bibfield  {journal} {\bibinfo  {journal} {SIAM Review}\ }\textbf {\bibinfo
  {volume} {56}},\ \bibinfo {pages} {385} (\bibinfo {year} {2014})}\BibitemShut
  {NoStop}%
\bibitem [{\citenamefont {Vico}\ \emph {et~al.}(2016)\citenamefont {Vico},
  \citenamefont {Greengard},\ and\ \citenamefont {Ferrando}}]{vico2016jcp}%
  \BibitemOpen
  \bibfield  {author} {\bibinfo {author} {\bibfnamefont {F.}~\bibnamefont
  {Vico}}, \bibinfo {author} {\bibfnamefont {L.}~\bibnamefont {Greengard}},\
  and\ \bibinfo {author} {\bibfnamefont {M.}~\bibnamefont {Ferrando}},\
  }\bibfield  {title} {\bibinfo {title} {Fast convolution with free-space
  {G}reen's functions},\ }\href {https://doi.org/10.1016/j.jcp.2016.07.028}
  {\bibfield  {journal} {\bibinfo  {journal} {J. Comput. Phys.}\ }\textbf
  {\bibinfo {volume} {323}},\ \bibinfo {pages} {191} (\bibinfo {year}
  {2016})}\BibitemShut {NoStop}%
\bibitem [{\citenamefont {Greengard}(1987)}]{greengard1987thesis}%
  \BibitemOpen
  \bibfield  {author} {\bibinfo {author} {\bibfnamefont {L.}~\bibnamefont
  {Greengard}},\ }\emph {\bibinfo {title} {Rapid evaluation of potential fields
  in particle systems}},\ \href@noop {} {Ph.D. thesis},\ \bibinfo  {school}
  {Yale University, New Haven, CT (USA)} (\bibinfo {year} {1987})\BibitemShut
  {NoStop}%
\bibitem [{\citenamefont {Greengard}(1988)}]{greengard1988}%
  \BibitemOpen
  \bibfield  {author} {\bibinfo {author} {\bibfnamefont {L.}~\bibnamefont
  {Greengard}},\ }\href {https://doi.org/10.7551/mitpress/5750.001.0001} {\emph
  {\bibinfo {title} {The rapid evaluation of potential fields in particle
  systems}}}\ (\bibinfo  {publisher} {MIT press},\ \bibinfo {year}
  {1988})\BibitemShut {NoStop}%
\bibitem [{\citenamefont {Greengard}\ and\ \citenamefont
  {Rokhlin}(1987)}]{greengard1987jcp}%
  \BibitemOpen
  \bibfield  {author} {\bibinfo {author} {\bibfnamefont {L.}~\bibnamefont
  {Greengard}}\ and\ \bibinfo {author} {\bibfnamefont {V.}~\bibnamefont
  {Rokhlin}},\ }\bibfield  {title} {\bibinfo {title} {A fast algorithm for
  particle simulations},\ }\href {https://doi.org/10.1016/0021-9991(87)90140-9}
  {\bibfield  {journal} {\bibinfo  {journal} {J. Comput. Phys.}\ }\textbf
  {\bibinfo {volume} {73}},\ \bibinfo {pages} {325} (\bibinfo {year}
  {1987})}\BibitemShut {NoStop}%
\bibitem [{\citenamefont {Cheng}\ \emph {et~al.}(1999)\citenamefont {Cheng},
  \citenamefont {Greengard},\ and\ \citenamefont {Rokhlin}}]{fmm2}%
  \BibitemOpen
  \bibfield  {author} {\bibinfo {author} {\bibfnamefont {H.}~\bibnamefont
  {Cheng}}, \bibinfo {author} {\bibfnamefont {L.}~\bibnamefont {Greengard}},\
  and\ \bibinfo {author} {\bibfnamefont {V.}~\bibnamefont {Rokhlin}},\
  }\bibfield  {title} {\bibinfo {title} {A fast adaptive multipole algorithm in
  three dimensions},\ }\href {https://doi.org/10.1016/0021-9991(87)90140-9}
  {\bibfield  {journal} {\bibinfo  {journal} {J. Comput. Phys.}\ }\textbf
  {\bibinfo {volume} {155}},\ \bibinfo {pages} {468} (\bibinfo {year}
  {1999})}\BibitemShut {NoStop}%
\bibitem [{\citenamefont {Gimbutas}\ and\ \citenamefont
  {Rokhlin}(2003)}]{gimbutas2003sisc}%
  \BibitemOpen
  \bibfield  {author} {\bibinfo {author} {\bibfnamefont {Z.}~\bibnamefont
  {Gimbutas}}\ and\ \bibinfo {author} {\bibfnamefont {V.}~\bibnamefont
  {Rokhlin}},\ }\bibfield  {title} {\bibinfo {title} {A generalized fast
  multipole method for nonoscillatory kernels},\ }\href
  {https://doi.org/10.1137/S1064827500381148} {\bibfield  {journal} {\bibinfo
  {journal} {SIAM J. Sci. Comput.}\ }\textbf {\bibinfo {volume} {24}},\
  \bibinfo {pages} {796} (\bibinfo {year} {2003})}\BibitemShut {NoStop}%
\bibitem [{\citenamefont {Greengard}\ and\ \citenamefont
  {Rokhlin}(1997)}]{fmm6}%
  \BibitemOpen
  \bibfield  {author} {\bibinfo {author} {\bibfnamefont {L.}~\bibnamefont
  {Greengard}}\ and\ \bibinfo {author} {\bibfnamefont {V.}~\bibnamefont
  {Rokhlin}},\ }\bibfield  {title} {\bibinfo {title} {A new version of the fast
  multipole method for the {L}aplace equation in three dimensions},\ }\href
  {https://doi.org/10.1017/S0962492900002725} {\bibfield  {journal} {\bibinfo
  {journal} {Acta Numer.}\ }\textbf {\bibinfo {volume} {6}},\ \bibinfo {pages}
  {229} (\bibinfo {year} {1997})}\BibitemShut {NoStop}%
\bibitem [{\citenamefont {Ying}\ \emph {et~al.}(2004)\citenamefont {Ying},
  \citenamefont {Biros},\ and\ \citenamefont {Zorin}}]{fmm7}%
  \BibitemOpen
  \bibfield  {author} {\bibinfo {author} {\bibfnamefont {L.}~\bibnamefont
  {Ying}}, \bibinfo {author} {\bibfnamefont {G.}~\bibnamefont {Biros}},\ and\
  \bibinfo {author} {\bibfnamefont {D.}~\bibnamefont {Zorin}},\ }\bibfield
  {title} {\bibinfo {title} {A kernel-independent adaptive fast multipole
  algorithm in two and three dimensions},\ }\href
  {https://doi.org/10.1016/j.jcp.2003.11.021} {\bibfield  {journal} {\bibinfo
  {journal} {J. Comput. Phys.}\ }\textbf {\bibinfo {volume} {196}},\ \bibinfo
  {pages} {591} (\bibinfo {year} {2004})}\BibitemShut {NoStop}%
\bibitem [{\citenamefont {Ethridge}\ and\ \citenamefont
  {Greengard}(2001)}]{ethridge2001sisc}%
  \BibitemOpen
  \bibfield  {author} {\bibinfo {author} {\bibfnamefont {F.}~\bibnamefont
  {Ethridge}}\ and\ \bibinfo {author} {\bibfnamefont {L.}~\bibnamefont
  {Greengard}},\ }\bibfield  {title} {\bibinfo {title} {{A New Fast-Multipole
  Accelerated {P}oisson Solver in Two Dimensions}},\ }\href
  {https://doi.org/10.1137/S1064827500369967} {\bibfield  {journal} {\bibinfo
  {journal} {SIAM J. Sci. Comput.}\ }\textbf {\bibinfo {volume} {23}},\
  \bibinfo {pages} {741} (\bibinfo {year} {2001})}\BibitemShut {NoStop}%
\bibitem [{\citenamefont {Malhotra}\ and\ \citenamefont
  {Biros}(2016)}]{malhotra2016toms}%
  \BibitemOpen
  \bibfield  {author} {\bibinfo {author} {\bibfnamefont {D.}~\bibnamefont
  {Malhotra}}\ and\ \bibinfo {author} {\bibfnamefont {G.}~\bibnamefont
  {Biros}},\ }\bibfield  {title} {\bibinfo {title} {Algorithm 967: A
  distributed-memory fast multipole method for volume potentials},\ }\href
  {https://doi.org/10.1145/2898349} {\bibfield  {journal} {\bibinfo  {journal}
  {ACM Trans. Math. Softw.}\ }\textbf {\bibinfo {volume} {43}},\ \bibinfo
  {pages} {1} (\bibinfo {year} {2016})}\BibitemShut {NoStop}%
\bibitem [{\citenamefont {Brandt}(1977)}]{brandt1977mcom}%
  \BibitemOpen
  \bibfield  {author} {\bibinfo {author} {\bibfnamefont {A.}~\bibnamefont
  {Brandt}},\ }\bibfield  {title} {\bibinfo {title} {Multi-level adaptive
  solutions to boundary-value problems},\ }\href
  {https://doi.org/10.2307/2006422} {\bibfield  {journal} {\bibinfo  {journal}
  {Mathematics of computation}\ }\textbf {\bibinfo {volume} {31}},\ \bibinfo
  {pages} {333} (\bibinfo {year} {1977})}\BibitemShut {NoStop}%
\bibitem [{\citenamefont {Hackbusch}(2013)}]{hackbusch2013}%
  \BibitemOpen
  \bibfield  {author} {\bibinfo {author} {\bibfnamefont {W.}~\bibnamefont
  {Hackbusch}},\ }\href {https://doi.org/10.1007/978-3-662-02427-0} {\emph
  {\bibinfo {title} {Multi-grid methods and applications}}},\ Vol.~\bibinfo
  {volume} {4}\ (\bibinfo  {publisher} {Springer Science \& Business Media},\
  \bibinfo {year} {2013})\BibitemShut {NoStop}%
\bibitem [{\citenamefont {Sampath}\ and\ \citenamefont
  {Biros}(2010)}]{sampath2010sisc}%
  \BibitemOpen
  \bibfield  {author} {\bibinfo {author} {\bibfnamefont {R.~S.}\ \bibnamefont
  {Sampath}}\ and\ \bibinfo {author} {\bibfnamefont {G.}~\bibnamefont
  {Biros}},\ }\bibfield  {title} {\bibinfo {title} {A parallel geometric
  multigrid method for finite elements on octree meshes},\ }\href
  {https://doi.org/10.1137/090747774} {\bibfield  {journal} {\bibinfo
  {journal} {SIAM J. Sci. Comput.}\ }\textbf {\bibinfo {volume} {32}},\
  \bibinfo {pages} {1361} (\bibinfo {year} {2010})}\BibitemShut {NoStop}%
\bibitem [{\citenamefont {Sundar}\ \emph {et~al.}(2015)\citenamefont {Sundar},
  \citenamefont {Stadler},\ and\ \citenamefont {Biros}}]{sundar2015nlaa}%
  \BibitemOpen
  \bibfield  {author} {\bibinfo {author} {\bibfnamefont {H.}~\bibnamefont
  {Sundar}}, \bibinfo {author} {\bibfnamefont {G.}~\bibnamefont {Stadler}},\
  and\ \bibinfo {author} {\bibfnamefont {G.}~\bibnamefont {Biros}},\ }\bibfield
   {title} {\bibinfo {title} {Comparison of multigrid algorithms for high-order
  continuous finite element discretizations},\ }\href
  {https://doi.org/10.1002/nla.1979} {\bibfield  {journal} {\bibinfo  {journal}
  {Numer. Linear Algebra Appl.}\ }\textbf {\bibinfo {volume} {22}},\ \bibinfo
  {pages} {664} (\bibinfo {year} {2015})}\BibitemShut {NoStop}%
\bibitem [{\citenamefont {Gholami}\ \emph {et~al.}(2016)\citenamefont
  {Gholami}, \citenamefont {Malhotra}, \citenamefont {Sundar},\ and\
  \citenamefont {Biros}}]{biros2016sisc}%
  \BibitemOpen
  \bibfield  {author} {\bibinfo {author} {\bibfnamefont {A.}~\bibnamefont
  {Gholami}}, \bibinfo {author} {\bibfnamefont {D.}~\bibnamefont {Malhotra}},
  \bibinfo {author} {\bibfnamefont {H.}~\bibnamefont {Sundar}},\ and\ \bibinfo
  {author} {\bibfnamefont {G.}~\bibnamefont {Biros}},\ }\bibfield  {title}
  {\bibinfo {title} {{FFT}, {FMM}, or multigrid? {A} comparative study of
  state-of-the-art {P}oisson solvers for uniform and nonuniform grids in the
  unit cube},\ }\href {https://doi.org/10.1137/15M1010798} {\bibfield
  {journal} {\bibinfo  {journal} {SIAM J. Sci. Comput.}\ }\textbf {\bibinfo
  {volume} {38}},\ \bibinfo {pages} {C280} (\bibinfo {year}
  {2016})}\BibitemShut {NoStop}%
\bibitem [{\citenamefont {Greengard}\ \emph {et~al.}(2024)\citenamefont
  {Greengard}, \citenamefont {Jiang}, \citenamefont {Rachh},\ and\
  \citenamefont {Wang}}]{jiang2024sirev}%
  \BibitemOpen
  \bibfield  {author} {\bibinfo {author} {\bibfnamefont {L.}~\bibnamefont
  {Greengard}}, \bibinfo {author} {\bibfnamefont {S.}~\bibnamefont {Jiang}},
  \bibinfo {author} {\bibfnamefont {M.}~\bibnamefont {Rachh}},\ and\ \bibinfo
  {author} {\bibfnamefont {J.}~\bibnamefont {Wang}},\ }\bibfield  {title}
  {\bibinfo {title} {A new version of the adaptive fast {G}auss transform for
  discrete and continuous sources},\ }\href
  {https://doi.org/10.1137/23M1572453} {\bibfield  {journal} {\bibinfo
  {journal} {SIAM Rev.}\ }\textbf {\bibinfo {volume} {66}},\ \bibinfo {pages}
  {287} (\bibinfo {year} {2024})}\BibitemShut {NoStop}%
\bibitem [{\citenamefont {Trefethen}(2008)}]{trefethen2008sirev}%
  \BibitemOpen
  \bibfield  {author} {\bibinfo {author} {\bibfnamefont {L.~N.}\ \bibnamefont
  {Trefethen}},\ }\bibfield  {title} {\bibinfo {title} {Is {G}auss quadrature
  better than {C}lenshaw--{C}urtis?},\ }\href
  {https://doi.org/10.1137/060659831} {\bibfield  {journal} {\bibinfo
  {journal} {SIAM Rev.}\ }\textbf {\bibinfo {volume} {50}},\ \bibinfo {pages}
  {67} (\bibinfo {year} {2008})}\BibitemShut {NoStop}%
\bibitem [{\citenamefont {Soler}\ \emph {et~al.}(2002)\citenamefont {Soler},
  \citenamefont {Artacho}, \citenamefont {Gale}, \citenamefont {Garc{\' i}a},
  \citenamefont {Junquera}, \citenamefont {Ordej{\' o}n},\ and\ \citenamefont
  {S{\' a}nchez-Portal}}]{SIESTA_Jose2002}%
  \BibitemOpen
  \bibfield  {author} {\bibinfo {author} {\bibfnamefont {J.~M.}\ \bibnamefont
  {Soler}}, \bibinfo {author} {\bibfnamefont {E.}~\bibnamefont {Artacho}},
  \bibinfo {author} {\bibfnamefont {J.~D.}\ \bibnamefont {Gale}}, \bibinfo
  {author} {\bibfnamefont {A.}~\bibnamefont {Garc{\' i}a}}, \bibinfo {author}
  {\bibfnamefont {J.}~\bibnamefont {Junquera}}, \bibinfo {author}
  {\bibfnamefont {P.}~\bibnamefont {Ordej{\' o}n}},\ and\ \bibinfo {author}
  {\bibfnamefont {D.}~\bibnamefont {S{\' a}nchez-Portal}},\ }\bibfield  {title}
  {\bibinfo {title} {The {SIESTA} method for ab initio order-{N} materials
  simulation},\ }\href {https://doi.org/10.1088/0953-8984/14/11/302} {\bibfield
   {journal} {\bibinfo  {journal} {J. Phys.: Condens. Matter}\ }\textbf
  {\bibinfo {volume} {14}},\ \bibinfo {pages} {2745} (\bibinfo {year}
  {2002})}\BibitemShut {NoStop}%
\bibitem [{\citenamefont {Blum}\ \emph {et~al.}(2009)\citenamefont {Blum},
  \citenamefont {Gehrke}, \citenamefont {Hanke}, \citenamefont {Havu},
  \citenamefont {Havu}, \citenamefont {Ren}, \citenamefont {Reuter},\ and\
  \citenamefont {Scheffler}}]{FHIaim_BLUM2009}%
  \BibitemOpen
  \bibfield  {author} {\bibinfo {author} {\bibfnamefont {V.}~\bibnamefont
  {Blum}}, \bibinfo {author} {\bibfnamefont {R.}~\bibnamefont {Gehrke}},
  \bibinfo {author} {\bibfnamefont {F.}~\bibnamefont {Hanke}}, \bibinfo
  {author} {\bibfnamefont {P.}~\bibnamefont {Havu}}, \bibinfo {author}
  {\bibfnamefont {V.}~\bibnamefont {Havu}}, \bibinfo {author} {\bibfnamefont
  {X.}~\bibnamefont {Ren}}, \bibinfo {author} {\bibfnamefont {K.}~\bibnamefont
  {Reuter}},\ and\ \bibinfo {author} {\bibfnamefont {M.}~\bibnamefont
  {Scheffler}},\ }\bibfield  {title} {\bibinfo {title} {Ab initio molecular
  simulations with numeric atom-centered orbitals},\ }\href
  {https://doi.org/https://doi.org/10.1016/j.cpc.2009.06.022} {\bibfield
  {journal} {\bibinfo  {journal} {Comput. Phys. Commun.}\ }\textbf {\bibinfo
  {volume} {180}},\ \bibinfo {pages} {2175} (\bibinfo {year}
  {2009})}\BibitemShut {NoStop}%
\bibitem [{\citenamefont {Lin}\ \emph {et~al.}(2012{\natexlab{a}})\citenamefont
  {Lin}, \citenamefont {Lu}, \citenamefont {Ying},\ and\ \citenamefont
  {E}}]{lin12}%
  \BibitemOpen
  \bibfield  {author} {\bibinfo {author} {\bibfnamefont {L.}~\bibnamefont
  {Lin}}, \bibinfo {author} {\bibfnamefont {J.}~\bibnamefont {Lu}}, \bibinfo
  {author} {\bibfnamefont {L.}~\bibnamefont {Ying}},\ and\ \bibinfo {author}
  {\bibfnamefont {W.}~\bibnamefont {E}},\ }\bibfield  {title} {\bibinfo {title}
  {Adaptive local basis set for {Kohn–Sham} density functional theory in a
  discontinuous {G}alerkin framework {I}: Total energy calculation},\ }\href
  {https://doi.org/https://doi.org/10.1016/j.jcp.2011.11.032} {\bibfield
  {journal} {\bibinfo  {journal} {J. Comput. Phys.}\ }\textbf {\bibinfo
  {volume} {231}},\ \bibinfo {pages} {2140} (\bibinfo {year}
  {2012}{\natexlab{a}})}\BibitemShut {NoStop}%
\bibitem [{\citenamefont {Kaye}\ \emph {et~al.}(2015)\citenamefont {Kaye},
  \citenamefont {Lin},\ and\ \citenamefont {Yang}}]{kaye15}%
  \BibitemOpen
  \bibfield  {author} {\bibinfo {author} {\bibfnamefont {J.}~\bibnamefont
  {Kaye}}, \bibinfo {author} {\bibfnamefont {L.}~\bibnamefont {Lin}},\ and\
  \bibinfo {author} {\bibfnamefont {C.}~\bibnamefont {Yang}},\ }\bibfield
  {title} {\bibinfo {title} {A posteriori error estimator for adaptive local
  basis functions to solve {Kohn–Sham} density functional theory},\ }\href
  {https://doi.org/10.4310/CMS.2015.v13.n7.a5} {\bibfield  {journal} {\bibinfo
  {journal} {Commun. Math. Sci.}\ }\textbf {\bibinfo {volume} {13}},\ \bibinfo
  {pages} {1741} (\bibinfo {year} {2015})}\BibitemShut {NoStop}%
\bibitem [{\citenamefont {Lin}\ \emph {et~al.}(2012{\natexlab{b}})\citenamefont
  {Lin}, \citenamefont {Lu}, \citenamefont {Ying},\ and\ \citenamefont
  {E}}]{lin15}%
  \BibitemOpen
  \bibfield  {author} {\bibinfo {author} {\bibfnamefont {L.}~\bibnamefont
  {Lin}}, \bibinfo {author} {\bibfnamefont {J.}~\bibnamefont {Lu}}, \bibinfo
  {author} {\bibfnamefont {L.}~\bibnamefont {Ying}},\ and\ \bibinfo {author}
  {\bibfnamefont {W.}~\bibnamefont {E}},\ }\bibfield  {title} {\bibinfo {title}
  {Optimized local basis set for {Kohn–Sham} density functional theory},\
  }\href {https://doi.org/https://doi.org/10.1016/j.jcp.2012.03.009} {\bibfield
   {journal} {\bibinfo  {journal} {J. Comput. Phys.}\ }\textbf {\bibinfo
  {volume} {231}},\ \bibinfo {pages} {4515} (\bibinfo {year}
  {2012}{\natexlab{b}})}\BibitemShut {NoStop}%
\bibitem [{\citenamefont {Hu}\ \emph {et~al.}(2015)\citenamefont {Hu},
  \citenamefont {Lin},\ and\ \citenamefont {Yang}}]{wei15}%
  \BibitemOpen
  \bibfield  {author} {\bibinfo {author} {\bibfnamefont {W.}~\bibnamefont
  {Hu}}, \bibinfo {author} {\bibfnamefont {L.}~\bibnamefont {Lin}},\ and\
  \bibinfo {author} {\bibfnamefont {C.}~\bibnamefont {Yang}},\ }\bibfield
  {title} {\bibinfo {title} {{DGDFT}: A massively parallel method for large
  scale density functional theory calculations},\ }\href
  {https://doi.org/10.1063/1.4931732} {\bibfield  {journal} {\bibinfo
  {journal} {J. Chem. Phys.}\ }\textbf {\bibinfo {volume} {143}},\ \bibinfo
  {pages} {124110} (\bibinfo {year} {2015})}\BibitemShut {NoStop}%
\bibitem [{\citenamefont {Sun}\ \emph {et~al.}(2020)\citenamefont {Sun},
  \citenamefont {Zhang}, \citenamefont {Banerjee}, \citenamefont {Bao},
  \citenamefont {Barbry}, \citenamefont {Blunt}, \citenamefont {Bogdanov},
  \citenamefont {Booth}, \citenamefont {Chen}, \citenamefont {Cui},
  \citenamefont {Eriksen}, \citenamefont {Gao}, \citenamefont {Guo},
  \citenamefont {Hermann}, \citenamefont {Hermes}, \citenamefont {Koh},
  \citenamefont {Koval}, \citenamefont {Lehtola}, \citenamefont {Li},
  \citenamefont {Liu}, \citenamefont {Mardirossian}, \citenamefont {McClain},
  \citenamefont {Motta}, \citenamefont {Mussard}, \citenamefont {Pham},
  \citenamefont {Pulkin}, \citenamefont {Purwanto}, \citenamefont {Robinson},
  \citenamefont {Ronca}, \citenamefont {Sayfutyarova}, \citenamefont
  {Scheurer}, \citenamefont {Schurkus}, \citenamefont {Smith}, \citenamefont
  {Sun}, \citenamefont {Sun}, \citenamefont {Upadhyay}, \citenamefont {Wagner},
  \citenamefont {Wang}, \citenamefont {White}, \citenamefont {Whitfield},
  \citenamefont {Williamson}, \citenamefont {Wouters}, \citenamefont {Yang},
  \citenamefont {Yu}, \citenamefont {Zhu}, \citenamefont {Berkelbach},
  \citenamefont {Sharma}, \citenamefont {Sokolov},\ and\ \citenamefont
  {Chan}}]{PySCF_2020}%
  \BibitemOpen
  \bibfield  {author} {\bibinfo {author} {\bibfnamefont {Q.}~\bibnamefont
  {Sun}}, \bibinfo {author} {\bibfnamefont {X.}~\bibnamefont {Zhang}}, \bibinfo
  {author} {\bibfnamefont {S.}~\bibnamefont {Banerjee}}, \bibinfo {author}
  {\bibfnamefont {P.}~\bibnamefont {Bao}}, \bibinfo {author} {\bibfnamefont
  {M.}~\bibnamefont {Barbry}}, \bibinfo {author} {\bibfnamefont {N.~S.}\
  \bibnamefont {Blunt}}, \bibinfo {author} {\bibfnamefont {N.~A.}\ \bibnamefont
  {Bogdanov}}, \bibinfo {author} {\bibfnamefont {G.~H.}\ \bibnamefont {Booth}},
  \bibinfo {author} {\bibfnamefont {J.}~\bibnamefont {Chen}}, \bibinfo {author}
  {\bibfnamefont {Z.-H.}\ \bibnamefont {Cui}}, \bibinfo {author} {\bibfnamefont
  {J.~J.}\ \bibnamefont {Eriksen}}, \bibinfo {author} {\bibfnamefont
  {Y.}~\bibnamefont {Gao}}, \bibinfo {author} {\bibfnamefont {S.}~\bibnamefont
  {Guo}}, \bibinfo {author} {\bibfnamefont {J.}~\bibnamefont {Hermann}},
  \bibinfo {author} {\bibfnamefont {M.~R.}\ \bibnamefont {Hermes}}, \bibinfo
  {author} {\bibfnamefont {K.}~\bibnamefont {Koh}}, \bibinfo {author}
  {\bibfnamefont {P.}~\bibnamefont {Koval}}, \bibinfo {author} {\bibfnamefont
  {S.}~\bibnamefont {Lehtola}}, \bibinfo {author} {\bibfnamefont
  {Z.}~\bibnamefont {Li}}, \bibinfo {author} {\bibfnamefont {J.}~\bibnamefont
  {Liu}}, \bibinfo {author} {\bibfnamefont {N.}~\bibnamefont {Mardirossian}},
  \bibinfo {author} {\bibfnamefont {J.~D.}\ \bibnamefont {McClain}}, \bibinfo
  {author} {\bibfnamefont {M.}~\bibnamefont {Motta}}, \bibinfo {author}
  {\bibfnamefont {B.}~\bibnamefont {Mussard}}, \bibinfo {author} {\bibfnamefont
  {H.~Q.}\ \bibnamefont {Pham}}, \bibinfo {author} {\bibfnamefont
  {A.}~\bibnamefont {Pulkin}}, \bibinfo {author} {\bibfnamefont
  {W.}~\bibnamefont {Purwanto}}, \bibinfo {author} {\bibfnamefont {P.~J.}\
  \bibnamefont {Robinson}}, \bibinfo {author} {\bibfnamefont {E.}~\bibnamefont
  {Ronca}}, \bibinfo {author} {\bibfnamefont {E.~R.}\ \bibnamefont
  {Sayfutyarova}}, \bibinfo {author} {\bibfnamefont {M.}~\bibnamefont
  {Scheurer}}, \bibinfo {author} {\bibfnamefont {H.~F.}\ \bibnamefont
  {Schurkus}}, \bibinfo {author} {\bibfnamefont {J.~E.~T.}\ \bibnamefont
  {Smith}}, \bibinfo {author} {\bibfnamefont {C.}~\bibnamefont {Sun}}, \bibinfo
  {author} {\bibfnamefont {S.-N.}\ \bibnamefont {Sun}}, \bibinfo {author}
  {\bibfnamefont {S.}~\bibnamefont {Upadhyay}}, \bibinfo {author}
  {\bibfnamefont {L.~K.}\ \bibnamefont {Wagner}}, \bibinfo {author}
  {\bibfnamefont {X.}~\bibnamefont {Wang}}, \bibinfo {author} {\bibfnamefont
  {A.}~\bibnamefont {White}}, \bibinfo {author} {\bibfnamefont {J.~D.}\
  \bibnamefont {Whitfield}}, \bibinfo {author} {\bibfnamefont {M.~J.}\
  \bibnamefont {Williamson}}, \bibinfo {author} {\bibfnamefont
  {S.}~\bibnamefont {Wouters}}, \bibinfo {author} {\bibfnamefont
  {J.}~\bibnamefont {Yang}}, \bibinfo {author} {\bibfnamefont {J.~M.}\
  \bibnamefont {Yu}}, \bibinfo {author} {\bibfnamefont {T.}~\bibnamefont
  {Zhu}}, \bibinfo {author} {\bibfnamefont {T.~C.}\ \bibnamefont {Berkelbach}},
  \bibinfo {author} {\bibfnamefont {S.}~\bibnamefont {Sharma}}, \bibinfo
  {author} {\bibfnamefont {A.~Y.}\ \bibnamefont {Sokolov}},\ and\ \bibinfo
  {author} {\bibfnamefont {G.~K.-L.}\ \bibnamefont {Chan}},\ }\bibfield
  {title} {\bibinfo {title} {{Recent developments in the PySCF program
  package}},\ }\href {https://doi.org/10.1063/5.0006074} {\bibfield  {journal}
  {\bibinfo  {journal} {J. Chem. Phys.}\ }\textbf {\bibinfo {volume} {153}},\
  \bibinfo {pages} {024109} (\bibinfo {year} {2020})}\BibitemShut {NoStop}%
\bibitem [{\citenamefont {Yeh}\ and\ \citenamefont
  {Morales}(2025)}]{CoQuiCode}%
  \BibitemOpen
  \bibfield  {author} {\bibinfo {author} {\bibfnamefont {C.-N.}\ \bibnamefont
  {Yeh}}\ and\ \bibinfo {author} {\bibfnamefont {M.}~\bibnamefont {Morales}},\
  }\href@noop {} {\bibinfo {title} {{C}o{Q}uí: {C}orrelated {Q}uantum
  {\'i}nterface}},\ \bibinfo {howpublished}
  {\url{https://github.com/AbInitioQHub/coqui/tree/main}} (\bibinfo {year}
  {2025})\BibitemShut {NoStop}%
\bibitem [{\citenamefont {Blackwell}\ \emph {et~al.}(2025)\citenamefont
  {Blackwell}, \citenamefont {Greengard}, \citenamefont {Jiang},\ and\
  \citenamefont {Malhotra}}]{DMKcode}%
  \BibitemOpen
  \bibfield  {author} {\bibinfo {author} {\bibfnamefont {R.}~\bibnamefont
  {Blackwell}}, \bibinfo {author} {\bibfnamefont {L.}~\bibnamefont
  {Greengard}}, \bibinfo {author} {\bibfnamefont {S.}~\bibnamefont {Jiang}},\
  and\ \bibinfo {author} {\bibfnamefont {D.}~\bibnamefont {Malhotra}},\
  }\href@noop {} {\bibinfo {title} {{DMK} {S}oftware {L}ibrary}},\ \bibinfo
  {howpublished} {\url{https://github.com/flatironinstitute/dmk}} (\bibinfo
  {year} {2025})\BibitemShut {NoStop}%
\bibitem [{\citenamefont {Bennett}\ \emph {et~al.}(2017)\citenamefont
  {Bennett}, \citenamefont {Melton}, \citenamefont {Annaberdiyev},
  \citenamefont {Wang}, \citenamefont {Shulenburger},\ and\ \citenamefont
  {Mitas}}]{ccecp_2ndrow_2017}%
  \BibitemOpen
  \bibfield  {author} {\bibinfo {author} {\bibfnamefont {M.~C.}\ \bibnamefont
  {Bennett}}, \bibinfo {author} {\bibfnamefont {C.~A.}\ \bibnamefont {Melton}},
  \bibinfo {author} {\bibfnamefont {A.}~\bibnamefont {Annaberdiyev}}, \bibinfo
  {author} {\bibfnamefont {G.}~\bibnamefont {Wang}}, \bibinfo {author}
  {\bibfnamefont {L.}~\bibnamefont {Shulenburger}},\ and\ \bibinfo {author}
  {\bibfnamefont {L.}~\bibnamefont {Mitas}},\ }\bibfield  {title} {\bibinfo
  {title} {{A new generation of effective core potentials for correlated
  calculations}},\ }\href {https://doi.org/10.1063/1.4995643} {\bibfield
  {journal} {\bibinfo  {journal} {J. Chem. Phys.}\ }\textbf {\bibinfo {volume}
  {147}},\ \bibinfo {pages} {224106} (\bibinfo {year} {2017})}\BibitemShut
  {NoStop}%
\bibitem [{\citenamefont {Annaberdiyev}\ \emph {et~al.}(2018)\citenamefont
  {Annaberdiyev}, \citenamefont {Wang}, \citenamefont {Melton}, \citenamefont
  {Bennett}, \citenamefont {Shulenburger},\ and\ \citenamefont
  {Mitas}}]{ccecp_tm_2018}%
  \BibitemOpen
  \bibfield  {author} {\bibinfo {author} {\bibfnamefont {A.}~\bibnamefont
  {Annaberdiyev}}, \bibinfo {author} {\bibfnamefont {G.}~\bibnamefont {Wang}},
  \bibinfo {author} {\bibfnamefont {C.~A.}\ \bibnamefont {Melton}}, \bibinfo
  {author} {\bibfnamefont {M.~C.}\ \bibnamefont {Bennett}}, \bibinfo {author}
  {\bibfnamefont {L.}~\bibnamefont {Shulenburger}},\ and\ \bibinfo {author}
  {\bibfnamefont {L.}~\bibnamefont {Mitas}},\ }\bibfield  {title} {\bibinfo
  {title} {{A new generation of effective core potentials from correlated
  calculations: 3d transition metal series}},\ }\href
  {https://doi.org/10.1063/1.5040472} {\bibfield  {journal} {\bibinfo
  {journal} {J. Chem. Phys.}\ }\textbf {\bibinfo {volume} {149}},\ \bibinfo
  {pages} {134108} (\bibinfo {year} {2018})}\BibitemShut {NoStop}%
\bibitem [{\citenamefont {Balabanov}\ and\ \citenamefont
  {Peterson}(2005)}]{ccpvtz_3d_2005}%
  \BibitemOpen
  \bibfield  {author} {\bibinfo {author} {\bibfnamefont {N.~B.}\ \bibnamefont
  {Balabanov}}\ and\ \bibinfo {author} {\bibfnamefont {K.~A.}\ \bibnamefont
  {Peterson}},\ }\bibfield  {title} {\bibinfo {title} {Systematically
  convergent basis sets for transition metals. {I}. {A}ll-electron correlation
  consistent basis sets for the 3d elements {Sc–Zn}},\ }\href
  {https://doi.org/10.1063/1.1998907} {\bibfield  {journal} {\bibinfo
  {journal} {J. Chem. Phys.}\ }\textbf {\bibinfo {volume} {123}},\ \bibinfo
  {pages} {064107} (\bibinfo {year} {2005})}\BibitemShut {NoStop}%
\bibitem [{\citenamefont {Pokhilko}\ \emph {et~al.}(2024)\citenamefont
  {Pokhilko}, \citenamefont {Yeh}, \citenamefont {Morales},\ and\ \citenamefont
  {Zgid}}]{THC_GWSOX_Pokhilko2024}%
  \BibitemOpen
  \bibfield  {author} {\bibinfo {author} {\bibfnamefont {P.}~\bibnamefont
  {Pokhilko}}, \bibinfo {author} {\bibfnamefont {C.-N.}\ \bibnamefont {Yeh}},
  \bibinfo {author} {\bibfnamefont {M.~A.}\ \bibnamefont {Morales}},\ and\
  \bibinfo {author} {\bibfnamefont {D.}~\bibnamefont {Zgid}},\ }\bibfield
  {title} {\bibinfo {title} {Tensor hypercontraction for fully self-consistent
  imaginary-time {GF2} and {GWSOX} methods: Theory, implementation, and role of
  the green’s function second-order exchange for intermolecular
  interactions},\ }\href {https://doi.org/10.1063/5.0215954} {\bibfield
  {journal} {\bibinfo  {journal} {J. Chem. Phys.}\ }\textbf {\bibinfo {volume}
  {161}},\ \bibinfo {pages} {084108} (\bibinfo {year} {2024})}\BibitemShut
  {NoStop}%
\bibitem [{\citenamefont {af~Klinteberg}\ \emph {et~al.}(2025)\citenamefont
  {af~Klinteberg}, \citenamefont {Greengard}, \citenamefont {Jiang},\ and\
  \citenamefont {Tornberg}}]{dmk_periodic_Ludvig2025}%
  \BibitemOpen
  \bibfield  {author} {\bibinfo {author} {\bibfnamefont {L.}~\bibnamefont
  {af~Klinteberg}}, \bibinfo {author} {\bibfnamefont {L.}~\bibnamefont
  {Greengard}}, \bibinfo {author} {\bibfnamefont {S.}~\bibnamefont {Jiang}},\
  and\ \bibinfo {author} {\bibfnamefont {A.-K.}\ \bibnamefont {Tornberg}},\
  }\href {https://arxiv.org/abs/2509.21471} {\bibinfo {title} {Fast summation
  of {S}tokes potentials using a new kernel-splitting in the {DMK} framework}}
  (\bibinfo {year} {2025}),\ \Eprint {https://arxiv.org/abs/2509.21471}
  {arXiv:2509.21471 [math.NA]} \BibitemShut {NoStop}%
\bibitem [{\citenamefont {Harbrecht}\ \emph {et~al.}(2012)\citenamefont
  {Harbrecht}, \citenamefont {Peters},\ and\ \citenamefont
  {Schneider}}]{iterative_cholesky_Harbrecht2012}%
  \BibitemOpen
  \bibfield  {author} {\bibinfo {author} {\bibfnamefont {H.}~\bibnamefont
  {Harbrecht}}, \bibinfo {author} {\bibfnamefont {M.}~\bibnamefont {Peters}},\
  and\ \bibinfo {author} {\bibfnamefont {R.}~\bibnamefont {Schneider}},\
  }\bibfield  {title} {\bibinfo {title} {On the low-rank approximation by the
  pivoted {C}holesky decomposition},\ }\href
  {https://doi.org/https://doi.org/10.1016/j.apnum.2011.10.001} {\bibfield
  {journal} {\bibinfo  {journal} {Appl. Numer. Math.}\ }\textbf {\bibinfo
  {volume} {62}},\ \bibinfo {pages} {428} (\bibinfo {year} {2012})}\BibitemShut
  {NoStop}%
\bibitem [{\citenamefont {Powell}(1981)}]{powell1981}%
  \BibitemOpen
  \bibfield  {author} {\bibinfo {author} {\bibfnamefont {M.~J.~D.}\
  \bibnamefont {Powell}},\ }\href {https://doi.org/10.1017/CBO9781139171502}
  {\emph {\bibinfo {title} {Approximation theory and methods}}}\ (\bibinfo
  {publisher} {Cambridge University Press},\ \bibinfo {year}
  {1981})\BibitemShut {NoStop}%
\end{thebibliography}%

\appendix

\section{Interpolative decomposition of pair densities\label{app:isdf}}

This appendix describes the algorithm used to construct the decomposition \eqref{eq:id_rho} of the pair density function $\rho_{ij}(\mathbf{r})$. We refer to Ref.~\onlinecite{THC-RPA_CNY2023} for further details.
We begin by defining the $N^{2}\times M$ pair density matrix 
\begin{align}
\rho_{kl} = \rho_{ij}(\bar{\mathbf{r}}_{l})
\end{align}
where the composite index $k=1,\ldots,N^2$ denotes a flattening of the orbital index pair $(i,j)$, and $\bar{\mathbf{r}}_{l}$, $l=1,\ldots,M$ is a real space grid sufficient to resolve all pair densities (which can be uniform or adaptive). 
The desired column-wise interpolative decomposition (ID)~\cite{ID_Cheng2005,ID_Liberty2007} of the matrix $\rho$ is given by 
\begin{align}
\rho_{kl} = \sum_{\mu=1}^R \rho_{k,l_\mu} \zeta_{\mu l}
\label{eq:over_complete_isdf}
\end{align}
for a subselection $l_\mu$ of $R$ grid points, which is equivalent to \eqref{eq:id_rho}. The $R$ interpolating points $\mathbf{r}_{\mu}$ correspond to the selected columns, $\mathbf{r}_\mu = \bar{\mathbf{r}}_{l_\mu}$.
The ISDF auxiliary basis functions $\zeta_\mu(\mathbf{r})$ can be defined via interpolation of the values $\zeta_{\mu}(\bar{\mathbf{r}}_{l}) = \zeta_{\mu l}$ (e.g., via piecewise Chebyshev interpolation in the case of the adaptive octree-based grid described in Sec.~\ref{subsec:adapt_poisson}). 
The ID \eqref{eq:over_complete_isdf} can be constructed in two steps: (i) select the interpolating points $\mathbf{r}_{\mu}$ from the original grid points $\bar{\mathbf{r}}_l$, yielding $\rho_{k, l_\mu}$, and (ii) solve the collection of $M$ overdetermined $N^2 \times R$ linear systems \eqref{eq:over_complete_isdf} to obtain $\zeta_{\mu l}$.

The standard method for column selection in the ID is to perform a rank-revealing column-pivoted QR factorization (QRCP) and to select the columns corresponding to the pivots~\cite{ID_Cheng2005}, but this procedure would be quartic-scaling in our case: $\OO{R M N^2}$. To avoid this, we consider the $M \times M$ pair density Gram matrix:
\begin{align}
  S_{l l'} = \sum_{k=1}^{N^2} \rho_{kl}\rho_{kl'}.
\end{align}
The interpolating points can then be selected corresponding to the pivots of a pivoted Cholesky factorization of this matrix~\cite{iterative_cholesky_Harbrecht2012}. This approach has the advantage that the Gram matrix $S$ never needs to be formed explicitly; only its diagonal elements and the $R$ rows corresponding to the pivots are required for an approximation of rank $R$. This leads to an overall $\OO{R^2 M}$ cost for selecting the interpolating points; see Ref.~\onlinecite{iterative_cholesky_Harbrecht2012} for details.
An alternative method is to use randomized QRCP algorithms on the matrix $\rho$, reducing the cost further to $\OO{MN^{2}\log{N}}$~\cite{ISDF_Lu2015} (recall $R > N$), but we have not explored this option in the present work.

Solving the collection of $N^2 \times R$ least squares problems \eqref{eq:over_complete_isdf} directly using QR factorization would have an $\OO{R M N^2}$ cost. To reduce this cost, we consider the normal equations for the least squares solutions:
\begin{align}
\sum_{\mu=1}^R A_{\nu \mu} \zeta_{\mu l} = Z_{\nu l}
\label{eq:thc_ls}
\end{align}
with 
\begin{equation} \label{eq:Z}
  \begin{aligned}
Z_{\mu l} &= \sum_{k=1}^{N^2} \rho_{k, l_\mu} \rho_{kl} \\
&= \sum_{i,j=1}^N \phi_i(\mathbf{r}_\mu) \phi_j(\mathbf{r}_\mu) \phi_i(\bar{\mathbf{r}}_l) \phi_j(\bar{\mathbf{r}}_l) \\
&= \left(\sum_{i=1}^N \phi_i(\mathbf{r}_\mu) \phi_i(\bar{\mathbf{r}}_l)\right)^2
  \end{aligned}
\end{equation}
and
\begin{equation} \label{eq:C} 
\begin{aligned}
A_{\mu\nu} &= \sum_{k=1}^{N^2} \rho_{k, l_\mu}\rho_{k, l_\nu} \\
&= \sum_{i,j=1}^N \phi_i(\mathbf{r}_\mu) \phi_j(\mathbf{r}_\mu) \phi_i(\mathbf{r}_\nu) \phi_j(\mathbf{r}_\nu) \\
&= \left(\sum_{i=1}^N \phi_i(\mathbf{r}_\mu) \phi_i(\mathbf{r}_\nu)\right)^2.
\end{aligned}
\end{equation}
Due to the separability of the pair densities, the costs of assembling $Z$ and $A$ are $\OO{R M N}$ and $\OO{R^2 N}$, respectively. The cost of solving \eqref{eq:thc_ls} directly is $\OO{R^3 + R^2 M}$. The dominant cost therefore scales as $\OO{R^2 M}$.

We note that \eqref{eq:thc_ls} gives an expression for the ISDF auxiliary basis functions as linear combinations of the pair densities,
\begin{align}
\zeta_{\mu}(\bar{\mathbf{r}}_l) =  \zeta_{\mu l} &= (A^{-1} Z)_{\mu l} \\
&= \sum_{\nu=1}^R \sum_{k=1}^{N^2} (A^{-1})_{\mu\nu} \rho_{k, l_\nu} \rho_{kl} \\
&= \sum_{i, j=1}^N \left(\sum_{\nu=1}^R (A^{-1})_{\mu\nu} \rho_{ij}(\mathbf{r}_\nu) \right) \rho_{ij}(\bar{\mathbf{r}}_l), 
\label{eq:zeta_expression}
\end{align}
suggesting that the same piecewise polynomial representation used for the pair densities might also be sufficient for the auxiliary basis functions. Further analysis would be required to establish such a statement rigorously.

\section{Approximation error of pair densities\label{app:pair_density_error}}
The following lemma can be found in \cite{powell1981}:
\begin{lemma}\label{lem:chebyinterbound}
  Let $l_i(x)$, $i=1,\ldots,n$, be the Lagrange polynomials
  for the interpolation nodes $x_i=\cos((i-0.5)\pi/n)$, the Chebyshev nodes
  of the first kind on $[-1,1]$. Let
  \be\label{eq:Lk1}
  L_n = \max_{x\in[-1,1]} \sum_{i=1}^n |l_i(x)|.
  \ee
  Then
  \be\label{eq:Lk2}
  L_n=\frac{1}{\pi}\sum_{i=1}^n\left|\cot\frac{(i-1/2)\pi}{2n}\right|.
  \ee
  Furthermore, let $f\in C^n[-1,1]$ and $p^C$ be its Chebyshev interpolating polynomial
  of degree less than $n$. Let $p^*$ be a best polynomial approximation of $f$ in $P$, the space
  of polynomials of degree less than $n$, i.e.,
  \be\label{eq:bestapprox}
  \|f-p^*\|_\infty \le \| f - p\|_\infty, \qquad \text{for any } p\in P.
  \ee
  Then
  \be\label{eq:chebapproxbound}
  \|f-p^C\|_\infty \le (1+L_n) \|f-p^*\|_\infty.
  \ee
\end{lemma}

\begin{remark}
  $L_n$ in \eqref{eq:Lk1} is the Lebesgue constant, which is the maximum norm of the Lagrange interpolation operator corresponding to the underlying interpolation nodes. For Chebyshev nodes,
  numerical calculations show that $L_n < 2.5$ for $n\le 20$. It can also be shown that
  \be
  L_n \le \frac{2}{\pi} \log n + 1.
  \ee
  It is straightforward to show that the maximum norm of the Lagrange interpolation operator in $d$ dimensions with tensor product Chebyshev interpolation nodes is simply
  $L_n^d$, and the three-dimensional version of \eqref{eq:chebapproxbound}
  is 
  \be\label{eq:chebapproxbound2}
  \|f-p^C\|_\infty \le (1+L_n^3) \|f-p^*\|_\infty,
  \ee
  where $p^C$ is the tensor-product Chebyshev interpolating polynomial with degree less than $n$ in each variable.
\end{remark}  

\begin{theorem}\label{thm:phiprodbound}
  Let $\{\phi_i\}_{i=1}^N$ be a set of continuous functions on $B_0=[-1,1]^3$. Let
  $p_i(\mathbf{r})=\sum_{jkl=0}^{n-1}C_{jkl}T_j(x_1)T_k(x_2)T_l(x_3)$
  be the tensor product Chebyshev interpolating polynomial of degree at most $n-1$ for $\phi_i(\mathbf{r})$, i.e.,
  $p_i(\mathbf{r}_\nu)=\phi_i(\mathbf{r}_\nu)$ with $\{\mathbf{r}_\nu\}_{\nu=1}^{n^3}$ the tensor product
  Chebyshev nodes on $B_0$. Suppose
  \be\label{eq:phierror}
  \|\phi_i-p_i\|_\infty \le \varepsilon, \qquad i=1,\ldots,N.
  \ee\\[3pt]
  Let $\{\tilde{\x}_\nu\}_{\nu=1}^{(2n)^3}$ be the $2n \times 2n \times 2n$
  tensor product Chebyshev nodes on $B_0$. Then
  the tensor product Chebyshev interpolating polynomial
  approximation of $\phi_i\phi_j$ at the nodes $\{\tilde{\x}_{\nu}\}$, denoted by
  $\tilde{\Phi}_{ij}$, satisfies the error bound
  \be\label{eq:phiprodbound}
  ||\phi_i\phi_j-\tilde{\Phi}_{ij}||_\infty\le (1+L_{2n}^3) (1+2C_\infty)\varepsilon,
  \ee
  where $C_\infty=\max_i||\phi_i||_\infty$. 
\end{theorem}
\begin{proof}
  Let $p_i$ denote the Chebyshev interpolating polynomial of $\phi_i$ of degree less than $n$.
  By definition, we have
  \be\label{eq:Phiij}
  \tilde{\Phi}_{ij}(\x) = \sum_{\mu=1}^{n^3} \phi_i(\x_\nu) \phi_j(\x_\nu) l_\nu(\x),
  \quad i,j=1,\ldots,N,
  \ee
  where $l_\nu(\x)$ is the tensor-product Lagrange basis polynomial equal to one
  at node $\nu$ and zero at all other nodes.
  Thus,
  \be\label{eq:producterror1}
  \ba
  \|\phi_i\phi_j-p_ip_j\|_\infty&=
  \|\phi_i\phi_j-\phi_ip_j + \phi_i p_j-p_ip_j\|_\infty\\
  &\le \|\phi_i(\phi_j-p_j)+(\phi_i-p_i)p_j\|_\infty\\
  &\le \|\phi_i\|_\infty\|\phi_j-p_j\|_\infty+\|p_j\|_\infty\|\phi_i-p_i\|_\infty\\
  &\le C_\infty\varepsilon + (C_\infty+\varepsilon)\varepsilon
  = (\varepsilon+2C_\infty)\varepsilon\\
  &<(1+2C_\infty)\varepsilon.
  \ea
  \ee
  Let $\tilde{\Phi}_{ij}$ be defined as above, and $P^*$ the best polynomial approximation of $\phi_i\phi_j$ of degree less than $2n$.
  By \eqref{eq:chebapproxbound2}, we have
  \be\label{eq:producterror2}
  \|\phi_i\phi_j-\tilde{\Phi}_{ij} \|_\infty\le (1+L_{2n}^3)\|\phi_i\phi_j-P^* \|_\infty.
  \ee
  Since $p_ip_j$ is a particular polynomial approximation of $\phi_i\phi_j$ of degree less than $2n$,
  we have
  \be\label{eq:producterror3}
  \|\phi_i\phi_j-P^* \|_\infty\le \|\phi_i\phi_j-p_ip_j \|_\infty.
  \ee
  Combining \eqref{eq:producterror1} -- \eqref{eq:producterror3}, we obtain
  \be\label{eq:producterror4}
  \ba
  \|\phi_i\phi_j-\tilde{\Phi}_{ij} \|_\infty&\le (1+L_{2n}^3)\|\phi_i\phi_j-P^* \|_\infty\\
  &\le (1+L_{2n}^3)\|\phi_i\phi_j-p_ip_j \|_\infty\\
  &\le (1+L_{2n}^3) (1+2C_\infty)\varepsilon.
  \ea
  \ee
\end{proof}

\begin{remark}
  The statement of Theorem~\ref{thm:phiprodbound} is given in the maximum norm. In practice, we use the relative
  $L^2$ norm to measure the approximation error, since the ERI tensor involves integrated quantities. In our numerical experiments, we observe that an upsampling factor of $1.5$ per dimension, rather than $2$, is sufficient to achieve the desired $L^2$ accuracy for the pair densities.
\end{remark}

\end{document}